\RequirePackage{fix-cm}
\documentclass[twocolumn]{svjour3}          
\smartqed  
\pdfoutput=1
\usepackage{amsmath,amssymb,amsfonts}
\usepackage{tikz}
\usetikzlibrary{intersections}
\usepackage{thmtools}
\usepackage{thm-restate}

\makeatletter
\let\MYcaption\@makecaption
\makeatother
\usepackage{subcaption}
\usepackage{graphicx}

\allowdisplaybreaks[1]

\newcommand{\eqdef}{\stackrel{\mbox{\scriptsize def}}{=}}

\newcommand{\reals}{\mathbb{R}}

\newcommand{\realsnng}{\mathbb{R}^{\ge0}}
\newcommand{\cala}{\mathcal{A}}

\newcommand{\cald}{\mathcal{D}}
\newcommand{\calf}{\mathcal{F}}

\newcommand{\calo}{\mathcal{O}}
\newcommand{\calp}{\mathcal{P}}
\newcommand{\calr}{\mathcal{R}}
\newcommand{\cals}{\mathcal{S}}

\newcommand{\calw}{\mathcal{W}}

\newcommand{\supp}{{\mathtt{supp}}}
\newcommand{\Dists}{\mathbb{D}} 

\newcommand{\randassign}{\ensuremath{\stackrel{\mathrm{\$}}{\leftarrow}}}
\newcommand{\expect}{\operatornamewithlimits{\mathbb{E}}}

\newcommand{\diverge}[2]{\mathit{D}(#1 \parallel #2)}

\newcommand{\renyidiverge}[3]{\mathit{D}_{#3}(#1 \parallel #2)}

\newcommand{\Dtv}[0]{\mathit{D}_{\sf tv}}
\newcommand{\tvdiverge}[2]{\renyidiverge{#1}{#2}{\sf tv}}

\newcommand{\cp}[2]{\mathsf{cp}(#1, #2)}
\newcommand{\utmetric}[0]{\mathit{d}}

\newcommand{\Winfu}{\mathit{W}_{\!\utmetric}}

\newcommand{\StatEL}{\mbox{\rm StatEL}}
\renewcommand{\phi} {\varphi}
\newcommand{\erightarrow}{\supset}

\newcommand{\hy}{\hat{y}}

\newcommand{\hv}{\hat{v}}
\newcommand{\hell}{\hat{\ell}}
\newcommand{\ellt}{\hell_{\mathsf{tar}}}
\newcommand{\tiv}{\widetilde{v}}

\newcommand{\Var}{\mathtt{Mes}}
\newcommand{\Label}{\mathtt{L}}

\newcommand{\sunny}{\mathit{sunny}}
\newcommand{\snowy}{\mathit{snowy}}

\newcommand{\man}{\eta_{\mathrm{m}}}
\newcommand{\woman}{\eta_{\mathrm{w}}}

\newcommand{\PR}[1]{\mathop{\mathbb{P}_{#1}}}

\newcommand{\MKa}{\mathop{\mathsf{K}_{a}}}

\newcommand{\MKeD}{\mathop{\mathsf{K}^{\varepsilon\!,\Winfu}}}

\newcommand{\MPa}{\mathop{\mathsf{P}_{\!a}}}

\newcommand{\MTDt}{\Delta_{T}}
\newcommand{\MTDx}[1]{\Delta_{#1}}
\newcommand{\trans}{\mathop{T}}
\newcommand{\IND}[2]{\mathbin{\sim}_{#1}^{#2}}

\newcommand{\M}{\mathfrak{M}}

\newcommand{\Rt}[0]{\calr_{T}}
\newcommand{\Rpsi}[0]{\calr_{\psi}}
\newcommand{\RpsiA}[0]{\calr_{\psi_0}}
\newcommand{\RpsiB}[0]{\calr_{\psi_1}}

\newcommand{\Rind}[0]{\calr_{x}^{\varepsilon,D}}
\newcommand{\Rrob}[0]{\calr_{x}^{\varepsilon,\!\Winfu}}

\newcommand{\wtrue}[0]{\mathit{w_{\,\sf true}}}
\newcommand{\wfin}[0]{\mathit{w_{\,\sf fin}}}
\newcommand{\wtrain}[0]{\mathit{w_{\,\sf train}}}
\newcommand{\wtest}[0]{\mathit{w_{\sf test}}}

\newcommand{\TP}{\mathit{tp}}
\newcommand{\TN}{\mathit{tn}}
\newcommand{\FP}{\mathit{fp}}
\newcommand{\FN}{\mathit{fn}}

\newcommand{\Loss}{\lambda_{L}}
\newcommand{\loss}{\mathit{loss}}
\newcommand{\Ltrain}{\mathit{L}_{\mathsf{train}}}

\newcommand{\GE}{\mathsf{GE}}

\newcommand{\Prevalence}{\mathsf{Prevalence}}
\newcommand{\Accuracy}{\mathsf{Accuracy}}
\newcommand{\Precision}{\mathsf{Precision}}
\newcommand{\FDR}{\mathsf{FDR}}
\newcommand{\FOR}{\mathsf{FOR}}
\newcommand{\NPV}{\mathsf{NPV}}
\newcommand{\Recall}{\mathsf{Recall}}
\newcommand{\FallOut}{\mathsf{FallOut}}
\newcommand{\MissRate}{\mathsf{MissRate}}
\newcommand{\Specificity}{\mathsf{Specificity}}

\newcommand{\TotalRobust}{\mathsf{Robust}}
\newcommand{\TargetRobust}{\mathsf{TRobust}}

\newcommand{\GrpFair}{\mathsf{GrpFair}}

\newcommand{\EqOdds}{\mathsf{EqOdds}}
\newcommand{\EqOpp}{\mathsf{EqOpp}}
\newcommand{\Suffice}{\mathsf{Sufficency}}

\newif\ifcommentson\commentsonfalse

\newif\ifconferenceon\conferenceontrue
\ifconferenceon
\newcommand{\arxiv}[1]{}
\newcommand{\conference}[1]{#1}
\newcommand{\conferenceShort}[1]{}
\else
\newcommand{\arxiv}[1]{#1}
\newcommand{\conference}[1]{}
\newcommand{\conferenceShort}[1]{}
\fi

\newcommand{\modified}[1]{\textcolor{black}{#1}}

\ifcommentson
\newcommand{\commentsize}[0]{.90\textwidth}
\newcommand{\commentYK}[1]{\begin{center} \parbox{\commentsize}{\textbf{\textcolor{black}{Comment Y.}} \textcolor{red}{#1} }\end{center}}
\newcommand{\replyYK}[1]{\begin{center} \parbox{\commentsize}{\textbf{Reply Y.} \textcolor{blue}{#1} }\end{center}}
\newcommand{\commentY}[1]{\marginpar{\footnotesize \color{red} {\bf Y:} \textsf{\scriptsize #1}}}
\newcommand{\replyY}[1]{\marginpar{\footnotesize \color{red} {\bf Y:} \textsf{\scriptsize #1}}}
\else
\newcommand{\commentYK}[1]{}
\newcommand{\replyYK}[1]{}
\newcommand{\commentY}[1]{}
\newcommand{\replyY}[1]{}
\fi

\newcommand{\pagelimitmarker}[1]{~\\ {\colorR{\ifthenelse{\thepage>#1}{\Huge Exceeding the page limit}{\huge Within the page limit}}}~\\ {\huge{\colorR{~~Page Limit\,\,\,\,\, = #1}}}~\\ {\huge{\colorR{~~Current Page = $\thepage$}}}}

\begin{document}

\title{
An Epistemic Approach to the Formal Specification of Statistical Machine Learning
\thanks{This work was supported 
by the New Energy and Industrial Technology Development Organization (NEDO), 
by ERATO HASUO Metamathematics for Systems Design Project (No. JPMJER1603), JST,
and by Inria under the project LOGIS.}
}

\author{Yusuke Kawamoto}

\institute{Yusuke Kawamoto \at
AIST, Tsukuba, JAPAN \\
~ \\
ORCID: 0000-0002-2151-9560
}

\date{Received: date / Accepted: date}

\maketitle

\begin{abstract}
We propose an epistemic approach to formalizing statistical properties of machine learning.
Specifically, we introduce a formal model for supervised learning based on a Kripke model where each possible world corresponds to a possible dataset and modal operators are interpreted as transformation and testing on datasets.
Then we formalize various notions of the classification performance, robustness, and fairness of statistical classifiers by using our extension of statistical epistemic logic (StatEL).
In this formalization, we show relationships among properties of classifiers, and relevance between classification performance and robustness.
As far as we know, this is the first work that uses epistemic models and logical formulas to express statistical properties of machine learning, and would be a starting point to develop theories of formal specification of machine learning.

\end{abstract}

\section{Introduction}
\label{sec:intro}
With the increasing use of machine learning in real-life applications, 
the safety and security of learning-based systems have been of great interest.
In particular, many recent studies~\cite{Szegedy:14:ICLR},\cite{Chakraborty:18:arxiv} have found vulnerabilities on the robustness of deep neural networks (DNNs) to malicious inputs, which can lead to disasters in security critical systems, such as self-driving cars.
To find out these vulnerabilities in advance, there have been researches on the formal verification and testing methods for the robustness of DNNs in recent years~\cite{Huang:17:CAV,Katz:17:CAV,Pei:17:SOSP,Tian:18:ICSE}.
However, relatively little attention has been paid to the formal specification of machine learning~\cite{Seshia:18:ATVA}.

In the research filed of formal specification and verification,
logical approaches have been shown useful to characterize desired properties and to develop theories to discuss those properties.
For example, temporal logic~\cite{Prior:1957} is a branch of modal logic for expressing time-dependent propositions, and has been widely used to describe requirements of hardware and software systems.
For another example, epistemic logic~\cite{vonWright:51:book} is a modal logic for knowledge and belief that has been employed as formal policy languages for distributed systems (e.g., for the authentication~\cite{Burrows:90:TOCS} and the anonymity \cite{Syverson:99:FM} of security protocols).
As far as we know, however, no prior work has employed logical formulas to rigorously describe various statistical properties of machine learning, although there are some papers that (often informally) list various desirable properties of machine learning~\cite{Seshia:18:ATVA}.

In this paper, we present a first logical formalization of statistical properties of machine learning.
To describe the statistical properties in a simple and abstract way, we extend \emph{statistical epistemic logic} (\StatEL{})~\cite{Kawamoto:19:FC}, which 
\modified{has recently been}
proposed to describe statistical knowledge and is applied to formalize statistical hypothesis testing and statistical privacy of databases.

A key idea in our modeling of statistical machine learning is that we formalize logical aspects in the syntax level, and statistical distances and dataset operations in the semantics level by using accessibility relations of a Kripke model~\cite{Kripke:63:MLQ}.
In this model, we formalize supervised learning and some of its desirable properties, including performance, robustness, and fairness.
More specifically, classification performance and robustness are described as the differences between the correct class label and the classifier's prediction, whereas fairness is expressed as a conditional indistinguishability between different groups.

\paragraph{Our contributions.}
The main contributions of this work are as follows:
\begin{itemize}
\item We propose a logical approach to formalizing statistical properties of machine learning in a simple and abstract way.
Specifically, we introduce a principle that logical aspects of statistical properties are described in the syntax level, and statistical distances and datasets are formalized in the semantics level.

\item We formalize supervised learning models and test datasets
\modified{(used to check whether the learning models satisfy specification)}
by employing a distributional Kripke model~\cite{Kawamoto:19:FC} where each possible world corresponds to a possible test dataset, 
and modal operators are interpreted as transformation and testing on datasets.
Then we show how 
the sampling from a dataset
and non-deterministic adversarial inputs are formalized in the distributional Kripke model.

\item We propose an extension of statistical epistemic logic (\StatEL{}) as a formal language to describe various properties of machine learning models,
including the performance, robustness, and fairness of statistical classifiers.
Then the satisfaction of logical formulas representing those properties is associated with their testing using a test dataset.
As far as we know, this is the first work that uses logical formulas to formalize various statistical properties of machine learning, and that provides an epistemic view on those properties.

\item We show some relationships among properties of classifiers, such as different levels of robustness.
We also present certain relationships between classification performance and robustness, which suggest robustness-related properties that have not been formalized in the literature as far as we know.
\end{itemize}

\paragraph{Cautions and limitations.}
In this paper, we focus on formalizing properties of supervised learning models that may be tested by using a dataset; i.e.,
we do \emph{not} deal with unsupervised learning, reinforcement learning, the properties of learning algorithms,
quality of training data (e.g., sample bias), quality of testing (e.g., coverage criteria), explainability, temporal properties, or system-level specification.
It should be noted that most of the properties formalized in this paper have been known in \modified{machine learning literatures}, and the novelty of this work lies in the logical formulation of those statistical properties.

We also \modified{highlight} that this work aims to provide a logical approach to the modeling of statistical properties tested with a dataset,
and does not present methods for checking, guaranteeing, or improving the performance\slash robustness\slash fairness of machine learning models.
As for the satisfiability of logical formulas,
we leave the development of testing and (statistical) model checking algorithms as future work, since the research area on the testing and verification of machine learning is relatively new and needs further techniques to improve the scalability.
Moreover, in some applications such as image recognition, some atomic formulas (e.g., representing whether an input image is \modified{a} panda) cannot be defined mathematically, and require additional techniques based on experiments.
Nevertheless, we demonstrate that describing various properties using logical formulas is useful to explore desirable properties and to discuss their relationships in a framework.

Finally, we emphasize that our work is the first attempt to use epistemic models and logical formulas to express statistical properties of machine learning models, and would be a starting point to develop theories of formal specification of machine learning in future research.

\paragraph{Relationship with the preliminary version.}
The main novelties of this paper with respect to the preliminary version~\cite{Kawamoto:19:SEFM} are as follows:
\begin{itemize}
\item We add how the satisfaction of a formula at a possible world can be regarded as the testing of a specification using a test dataset (Sect.~\ref{sub:testing:dataset}).
\item We show how modal operators are used to model the transformation and testing on datasets.
For example, \emph{data preparation} $\trans$ (e.g., data cleaning, data augmentation) can also be formalized as a modal operator $\MTDt$  (Sect.~\ref{sub:modality:transform}).
\item We re-interpret the non-classical implication $\erightarrow$ for conditional probabilities in \StatEL{} as a modal operator associated with a conditioning relation (Sect.~\ref{sub:modal:conditioning}).
\item We introduce a modal operator $\IND{x}{\varepsilon,D}$ for conditional indistinguishability (Sect. \ref{sub:modal:indistinguishable}).
Then we provide a more comprehensible formalization of the fairness of supervised learning (Sect.~\ref{sec:ML:fairness}) without using counterfactual epistemic operators~\cite{Kawamoto:19:SEFM},
\modified{because the formalization using these operators requires an additional formula and makes the presentation more complicated and unintuitive}.
\item We add a formalization of \emph{generalization error} to capture how accurately a classifier is able to classify previously unseen input data (Sect.~\ref{sub:ML:correctness:generalize}).
\item We add a formalization of other fairness notions called \emph{separation} (Sect.~\ref{sub:fair:seperation}) and \emph{sufficiency} (Sect.~\ref{sub:fair:sufficiency}) so that this paper covers all three categories of fairness notions~\cite{Barocas:19:book}.
\item \modified{We show a running example of a pedestrian detection to illustrate the formalization of various notions of performance, robustness, and fairness.}
\end{itemize}

\paragraph*{Paper organization.}~
The rest of this paper is organized as follows.
Sect.~\ref{sec:preliminaries} presents notations used in this paper and provides background on statistical distances and statistical epistemic logic (\StatEL{}).
Sect.~\ref{sec:extend:StatEL} introduces a different view on the modal operators in \StatEL{} and extends the logic with additional operators.
Sect.~\ref{sec:formal:ML} introduces a formal model for describing the behaviors of statistical classifiers and non-deterministic adversarial inputs.
Sects.~\ref{sec:ML:performance-prediction}, \ref{sec:ML:robustness}, and~\ref{sec:ML:fairness} respectively formalize various notions of the performance, robustness, and fairness of classifiers by using our extension of \StatEL{}.
Sect.~\ref{sec:related} presents related work and Sect.~\ref{sec:conclude} concludes.

\section{Preliminaries}
\label{sec:preliminaries}
In this section we introduce some notations, and review background on statistical distance notions and the syntax and semantics of \emph{statistical epistemic logic} (\StatEL{}), introduced in~\cite{Kawamoto:19:FC}.

\subsection{Notations}
\label{sub:notations}

Let $\realsnng$ be the set of non-negative real numbers,
and $[0, 1]$ be the set of non-negative real numbers not greater than~$1$.
We denote by $\Dists\calo$ the set of all probability distributions over a finite set~$\calo$. Given a finite set $\calo$ and a probability distribution $\mu\in\Dists\calo$, the probability of sampling a value $v$ from $\mu$ is denoted by $\mu[v]$.
For a subset $R\subseteq\calo$, let $\mu[R] = \sum_{v\in R} \mu[v]$.
For a distribution $\mu$ over a finite set $\calo$, its \emph{support} is defined by 
$\supp(\mu) = \{ v \in \calo \colon \mu[v] > 0 \}$.

\subsection{Statistical Distance}
\label{sub:stat-distance}

We recall popular notions of distance between probability distributions: \emph{total variation} and \emph{$\infty$-Wasserstein distance}.

\modified{Informally, total variation between two distributions $\mu_0$ and $\mu_1$ over a set $\calo$ represents the largest difference between the probabilities that $\mu_0$ and $\mu_1$ assign to an identical subset $R$ of $\calo$.}

\begin{definition}[Total variation]\label{def:TV}\rm
For a finite set $\calo$, the \emph{total variation} $\Dtv$ of two distributions $\mu_0, \mu_1 \in \Dists\calo$ is defined by:
\begin{align*}
\tvdiverge{\mu_0}{\mu_1} \eqdef
\sup_{R \subseteq \calo} | \mu_0[R] - \mu_1[R] |
{.}
\end{align*}
\end{definition}

\modified{We then recall the $\infty$-Wasserstein metric~\cite{Vaserstein:69:PPI}.
Intuitively, the $\infty$-Wasserstein metric $\Winfu(\mu_0, \mu_1)$ between two distributions $\mu_0, \mu_1$ represents the minimum largest move between points in a transportation from $\mu_0$ to $\mu_1$.}

\begin{definition}[$\infty$-Wasserstein metric]
\label{def:p-Wasserstein-metric}\rm
Let $\calo$ be a finite set and $\utmetric: \calo\times\calo \rightarrow \realsnng$ be a metric over $\calo$.
The \emph{$\infty$-Wasserstein metric} $\Winfu$ w.r.t. $\utmetric$ between two distributions $\mu_0, \mu_1\in\Dists\calo$ is defined~by:
\[
\Winfu(\mu_0, \mu_1) = 
\min_{\mu\in \cp{\mu_0}{\mu_1}}\hspace{-3ex}
\max_{\hspace{3ex}(v_0, v_1)\in\supp(\mu)}\hspace{-2ex}
\utmetric(v_0, v_1)
\]
where $\cp{\mu_0}{\mu_1}$ is the set of all couplings%
\footnote{A \emph{coupling} of two distributions $\mu_0, \mu_1\in\Dists\calo$ is a joint distribution $\mu\in\Dists(\calo\times \calo)$ such that $\mu_0$ and $\mu_1$ are $\mu$'s marginal distributions, i.e.,
for each $v_0\in \calo$,
$\mu_0[v_0] = \sum_{v'_1\in \calo} \mu[v_0, v'_1]$ and
for each $v_1\in \calo$,
$\mu_1[v_1] = \sum_{v'_0\in \calo} \mu[v'_0, v_1]$.
\modified{For a coupling $\mu$, the support $\supp(\mu)$ is the maximum subset of $\calo\times\calo$ whose elements are assigned non-zero probabilities in~$\mu$.}
}
of $\mu_0$ and~$\mu_1$.
\end{definition}

\subsection{Syntax of \StatEL{}}
\label{sub:syntax}

We next recall the syntax of statistical epistemic logic (\StatEL{}) \cite{Kawamoto:19:FC}, 
which has two levels of formulas: \emph{static} and \emph{epistemic formulas}.
Intuitively, a static formula describes a proposition satisfied at a (deterministic) state, while an epistemic formula describes a proposition satisfied at a probability distribution of states.
In this paper, the former is used only to define the latter.

Formally, let $\Var$ be a set of symbols called \emph{measurement variables}, and
$\Gamma$ be a set of atomic formulas of the form $\gamma(x_1, x_2, \ldots, x_n)$ for a predicate symbol $\gamma$, $n \ge 0$, and $x_1, x_2, \ldots, x_n\in\Var$.
Let $I \subseteq [0, 1]$ be a finite union of disjoint intervals, and $\cala$ be a finite set of indices (e.g., associated with statistical divergences).
Then the formulas are defined by:
\begin{align*}
&\mbox{Static formulas:}~
\psi \mathbin{::=}
 \gamma(x_1, x_2, \ldots, x_n) \,|\,
 \neg \psi \,|\, \psi \wedge \psi
\\
&\mbox{Epistemic formulas:}~
\phi \mathbin{::=}
\PR{I} \psi \,|\, \neg \phi \,|\, \phi \wedge \phi \,|\,
\psi {\erightarrow} \phi \,|\, {\MKa}\phi
\end{align*}
where 
$a\in\cala$.
We denote by $\calf$ the set of all epistemic formulas.
Note that we have no quantifiers over measurement variables. 
(See Sect.~\ref{sub:interpretation} for more details.)

The \emph{probability quantification} $\PR{I} \psi$ represents that a static formula $\psi$ is satisfied with a probability belonging to a set $I$.
For instance, $\PR{(0.95, 1]} \psi$ represents that $\psi$ holds with a probability greater than $0.95$.
By $\psi \erightarrow \PR{I} \psi'$ we represent that the conditional probability of $\psi'$ given $\psi$ is included in a set $I$.
The \emph{epistemic knowledge} $\MKa \phi$ expresses that 
we \modified{know} $\phi$
\modified{when our capability of observation is denoted by $a\in\cala$}.

As syntax sugar, we use \emph{disjunction} $\vee$, \emph{classical implication} $\rightarrow$, and \emph{epistemic possibility} $\MPa$, defined as usual by:
$\phi_0 \vee \phi_1 \mathbin{::=} \neg (\neg \phi_0 \wedge \neg \phi_1)$,
$\phi_0 \rightarrow \phi_1 \mathbin{::=} \neg \phi_0 \vee \phi_1$,
and $\MPa{\phi} \mathbin{::=} \neg \MKa \neg \phi$.
When $I$ is a singleton $\{ i \}$, we abbreviate $\PR{I}$ as $\PR{i}$.

\subsection{Distributional Kripke Model}
\label{sub:Kripke}

Next we recall the notion of a distributional Kripke model \cite{Kawamoto:19:FC}, where each possible world is associated with a probability distribution over a set of states, and with a stochastic assignment of data to measurement variables.

\begin{definition}[Distributional Kripke model] \label{def:dist-Kripke-model} \rm
Let $\cala$ be a finite set of indices (typically associated with operations and tests on datasets),
$\cals$ be a finite set of states, and $\calo$ be a finite set of data, called a \emph{data domain}.
A \emph{distributional Kripke model} 
is a tuple 
$\M =(\calw, (\calr_a)_{a\in\cala}, (V_s)_{s\in\cals})$ 
consisting of:
\begin{itemize}
\item a non-empty set $\calw$ of multisets of states belonging to $\cals$;
\item for each $a\in\cala$, an accessibility relation $\calr_a \subseteq \calw {\times}\!\calw$;
\item for each $s\in\cals$, a valuation $V_s: \Gamma\rightarrow\calp(\calo^k)$ that maps each $k$-ary predicate $\gamma$ to a set $V_s(\gamma)$ of $k$-tuples of data.
\end{itemize}
The set $\calw$ is called a \emph{universe}, and its elements are called \emph{possible worlds}.
A world is said to be \emph{finite} if it is a finite multiset, i.e., it has a finite number of (possibly duplicated) elements.
\modified{A world is said to be \emph{infinite} if it is an infinite multiset.}
\end{definition}

\modified{
The relation $\calr_a$ determines an accessibility between two worlds.
For example, $(w, w')\in\calr_a$ means that a world $w'$ is accessible from a world~$w$ when our capability of distinguishing possible worlds is denoted by $a\in\cala$.
The valuation $V_s$ may give a possibly different interpretation of a predicate $\gamma$ at a different state $s$.
}%
We assume that all measurement variables range over the same data domain $\calo$ in every world.
\modified{The interpretation of measurement variables at a state $s$ is given by a deterministic assignment $\sigma_s$ defined below.}

\begin{definition}[Deterministic assignment] \label{def:deterministic-assignment} \rm
For any distributional Kripke model $\M {\,=}(\calw\!, (\calr_a)_{a\in\cala}, (V_s)_{s\in\cals})$, 
we assume that each world $w\in\calw$ is associated with a function $\rho_w: \Var\times\cals\rightarrow\calo$ that maps each measurement variable $x$ to its value $\rho_w(x, s)$ that is observed at a state~$s$ belonging to the world $w$.
We also assume that each state $s$ in a world $w$ is associated with the \emph{deterministic assignment} $\sigma_s: \Var\rightarrow\calo$ defined by $\sigma_s(x) = \rho_w(x, s)$.
\end{definition}

Since each world $w$ is a multiset of states, we abuse the notation and denote by $w[s]$ the probability that a state $s$ is randomly chosen from $w$ (i.e., the number of occurrences of $s$ in the multiset $w$, divided by the total number of elements in $w$).
Here we regard each world $w$ as a probability distribution over the states that corresponds to the multiset.

The probability that a measurement variable $x\in\Var$ has a value $v\in\calo$ is:
$\sigma_w(x)[v] = \sum_{\substack{s\in w, \sigma_s(x) = v}} w[s]$.
Note that $\sigma_w: \Var\rightarrow\Dists\calo$ maps each measurement variable $x$ to a probability distribution $\sigma_w(x)$ over the data domain~$\calo$. 
Hence $\sigma_w$ represents the joint probability distribution of all variables in $\Var$, 
and is called the \emph{stochastic assignment} at $w$.
When a state $s$ is uniformly drawn from a multiset $w$ of states, a datum $\sigma_s(x)$ is sampled from the distribution $\sigma_w(x)$.

In later sections, a possible world corresponds to a \emph{dataset} (i.e., a multiset of data tuples) from which data are sampled.
For example, suppose that we have only three measurement variables $\Var = \{ x, y, z \}$.
Then for each state $s$ in a world $w$, the deterministic assignment $\sigma_{s}: \Var\rightarrow\calo$ represents the tuple of data $(\sigma_{s}(x), \sigma_{s}(y), \sigma_{s}(z))$.
Hence each state $s$ corresponds to a tuple of data, and the world $w$ corresponds to the dataset $\{ (\sigma_{s}(x), \sigma_{s}(y), \sigma_{s}(z)) \mid s\in w \}$.

\subsection{Stochastic Semantics of \StatEL{}}
\label{sub:interpretation}

Now we recall the \emph{stochastic semantics}\,\cite{Kawamoto:19:FC} for the \StatEL{} formulas over a distributional Kripke model 
$\M =(\calw, \allowbreak (\calr_a)_{a\in\cala}, (V_s)_{s\in\cals})$ 
with $\calw = \Dists\cals$.

The interpretation of a static formulas $\psi$ at a state $s$ is given by:
\begin{align*}
s \models \gamma(x_1, \ldots, x_k) 
& ~\mbox{ iff }~
(\sigma_s(x_1), \ldots, \sigma_s(x_k)) \in V_s(\gamma)
\\
s \models \neg \psi 
& ~\mbox{ iff }~
s \not\models \psi 
\\
s \models \psi \wedge \psi'
& ~\mbox{ iff }~
s \models \psi
~\mbox{ and }~
s \models \psi'
{.}
\end{align*}

The \emph{restriction} $w|_\psi$ of a world $w$ to a static formula $\psi$ is defined by
$w|_\psi[s] = \frac{w[s]}{\sum_{s': s' \models \psi} w[s']}$ if $s \models \psi$,
and $w|_\psi[s] = 0$ otherwise.
Note that $w|_\psi$ is undefined if there is no state $s$ that satisfies $\psi$ and has a non-zero probability in $w$.

Then the interpretation of epistemic formulas in a world $w$ is defined by:
\begin{align*}
\M, w \models \PR{I} \psi
& ~\mbox{ iff }~
\Pr\!\left[ s \randassign w :~ s \models \psi \right] \in I
\\
\M, w \models \neg \phi
& ~\mbox{ iff }~
\M, w \not\models \phi
\\
\M, w \models \phi \wedge \phi'
& ~\mbox{ iff }~
\M, w \models \phi
~\mbox{ and }~
\M, w \models \phi'
\\
\M, w \models \psi \erightarrow \phi
& ~\mbox{ iff }~
\mbox{$w|_{\psi}$ is defined and }~
\M, w|_{\psi} \models \phi
\\
\M, w \models \MKa \phi
& ~\mbox{ iff }~
\mbox{for every $w'$ s.t. $(w, w') \in \calr_a$, }~
\\& ~\hspace{3.5ex}~
\M, w' \models \phi
{,}
\end{align*}
where $s \randassign w$ represents that a state $s$ is sampled from the distribution $w$.

Then $\M, w \models \psi_0 \erightarrow \PR{I} \psi_1$ represents that 
the conditional probability of satisfying a static formula $\psi_1$ given another $\psi_0$ is included in a set $I$ at a world~$w$.

In each world $w$, measurement variables can be interpreted using $\sigma_w$. 
This allows us to assign different values to different occurrences of a variable in a formula;
E.g., in $\phi(x) \rightarrow \MKa \phi'(x)$,\, $x$ occurring in $\phi(x)$ is interpreted by $\sigma_{w}$ in a world $w$, while $x$ in $\phi'(x)$ is interpreted by $\sigma_{w'}$ in another $w'$ s.t. $(w, w')\in \calr_a$.

Finally, the interpretation of an epistemic formula $\phi$ in $\M$ is given by:
\begin{align*}
\M \models \phi
& ~\mbox{ iff }~
\mbox{for every world $w$ in $\M$, }~
\M, w \models \phi
{.}
\end{align*}

Hereafter we mainly focus on the satisfaction local to a possible world, and $\M$ may be omitted when it is clear from the context.

\section{Modality as Transformation and Testing on Datasets}
\label{sec:extend:StatEL}

In this section we introduce a different view on the modal operators in statistical epistemic logic (\StatEL{}), and define additional modal operators that are used to formalize various properties of machine learning in Sects.~\ref{sec:ML:performance-prediction} to~\ref{sec:ML:fairness}.

\subsection{Checking Satisfaction at a World as Testing with a Dataset}
\label{sub:testing:dataset}

We first show how \modified{we regard the satisfaction of a formula $\phi$ as testing a system's specification expressed by $\phi$} as follows.

As explained in Sect.~\ref{sub:Kripke}, a possible world corresponds to a possible dataset.
Thus, given a model $\M$, a world $w$, and a formula $\phi$, checking the satisfaction $\M, w \models \phi$ can be regarded as testing whether the specification $\phi$ \modified{of a system 
(e.g., a machine learning model we formalize in Sect.~\ref{sec:formal:ML})} is satisfied 
\modified{when the dataset~$w$ provides inputs to the system}.
\modified{For example, let $\phi$ be a formula representing that a machine learning task (e.g., classification) $C$ fails with probability at most $5\%$.
Then  $\M, w \models \phi$ represents that when the learning task $C$ is performed using a test dataset $w$, then it fails for at most $5\%$ of the test data in $w$.}

For simplicity, we discuss the satisfaction of the formulas $\phi$ in which neither $\MKa$ nor $\MPa$ occurs as follows.
For each state (namely, data tuple) $s\in w$ and for each static sub-formula $\psi$ of $\phi$, we can efficiently check whether $s \models \psi$.

When the dataset $w$ is finite (i.e., it is a finite multiset of data tuples), we can check the satisfaction $w \models \phi$ in finite time,
more precisely, in linear time in the number of elements in~$w$.

When the dataset $w$ is infinite, however, we cannot check whether $w \models \phi$ in general.
For example, suppose that $w$ \modified{is} the infinite dataset representing a true distribution from which data are sampled and observed.
When we cannot learn $w$ itself, we usually obtain a finite dataset $\wfin$ by sampling data from $w$ repeatedly and independently and check a specification $\phi$ only with this test dataset $\wfin$.

Hereafter, we mainly deal with distributional Kripke models $\M$ that have infinite numbers of finite worlds.
In the following sections except Sect.~\ref{sec:ML:robustness}, we deal only with formulas without $\MKa$ nor $\MPa$,\,%
\footnote{
The testing of a formula $\phi$ is not feasible when an epistemic operator $\MKa$ or $\MPa$ occurs in $\phi$ and the model $\M$ has a large number of possible worlds.
Detailed analysis of time complexity of \StatEL{} is out of the scope of this paper, and should be included in the journal version of our paper~\cite{Kawamoto:19:FC} that proposed \StatEL{}.
As we will discuss in Sect.~\ref{sec:ML:robustness}, the robustness of machine learning is formalized using these epistemic operators, hence cannot be tested in practical time unless $\M$ is comprised of a small number of worlds.} 
hence can check their satisfaction at a finite world in finite time.

\subsection{Modal Operators for Dataset Transformation}
\label{sub:modality:transform}

In the rest of Sect.~\ref{sec:extend:StatEL}, we show that modal operators can be used to model the transformation and testing on datasets.

First, we introduce \emph{modal operators for dataset transformation}.
The modal operator $\MTDt$ defined below is unary (i.e., taking a single formula as argument), and is parameterized with a transformation $\trans$ between datasets.
Intuitively, $w \models \MTDt \phi$
represents that a formula $\phi$ is satisfied for the dataset $w'$ that is obtained by transforming the current dataset $w$ by $\trans$.
Formally, the modal operator $\MTDt$ is interpreted as follows.

\begin{definition}[Modality $\MTDt$ for a dataset transformation $T$]\label{def:modal:transform} \rm
Given a function $T: \calw\rightarrow\calw$, we define an accessibility relation as 
$\Rt \eqdef \{ (w, w') \mid w' = T(w) \}$.
Then we define the interpretation of $\MTDt$ by:
\begin{align*}
&~~\M, w \models \MTDt \phi
\\&
~\mbox{ iff }~
\mbox{there is a $w'$ s.t. $(w, w') \in \Rt$ and }~
\M, w' \models \phi
{.}
\end{align*}
\end{definition}

For example, machine learning often require \emph{data preparation} to manipulate a given raw dataset into a form that makes a machine learning task feasible and more effective 
(e.g., \emph{data cleaning}, \emph{data augmentation}).
\modified{For a dataset $w$ and two ways of data preparation $\trans_0$ and $\trans_1$,\, $w \models \MTDx{\trans_0} \phi \wedge \MTDx{\trans_1} \phi$
represents that a property $\phi$ holds 
for the two prepared datasets $\trans_0(w)$ and $\trans_1(w)$.}

For another example, the security of machine learning often assumes a certain malicious adversary that can manipulate a given dataset to make a machine learning task fail.
Such adversarial operations $\trans$ on \allowbreak datasets can also be formalized using a different modal operator corresponding to $\trans$ as we will explain in Sect.~\ref{sec:ML:robustness}.

\modified{In the next section, we show that the logical connective $\erightarrow$ can be re-interpreted as the modality $\MTDt$ for some dataset transformation $T$.}

\subsection{Modality for Conditioning}
\label{sub:modal:conditioning}

We then present another interpretation of the logical connective $\erightarrow$ (defined in Sect.~\ref{sub:interpretation}) used to express conditional probabilities 
\modified{in Sects.~\ref{sec:ML:performance-prediction} and~\ref{sec:ML:robustness}}.
Roughly speaking, we regard the restriction $w|_{\psi}$ of a world $w$ to a static formula $\psi$ as a transformation $\Rpsi$ of $w$.
Then we redefine $\erightarrow$ as a modal operator associated with $\Rpsi$, and call it the \emph{conditioning operator}.
Formally, the interpretation of $\erightarrow$ is defined as follows.

\begin{definition}[Conditioning operator $\erightarrow$]\label{def:modal:conditioning} \rm
Assume that the universe $\calw$ includes all sub-multisets of each $w\in\calw$.
Given a static formula $\psi$, we define an accessibility relation as the \emph{conditioning relation}
$\Rpsi \eqdef \{ (w, w|_{\psi}) \mid w\in\calw \}$.
Then the interpretation of the conditioning operator $\erightarrow$ is given by:
\begin{align*}
&~~\M, w \models \psi \erightarrow \phi
\\ &~\mbox{ iff }~
\mbox{there is a $w'$ s.t. $(w, w')\in\Rpsi$ and }~
\M, w' \models \phi
{.}
\end{align*}
\end{definition}

Intuitively, $w \models \psi \erightarrow \phi$ corresponds to the two operations: 
(i) transforming the given dataset $w$ to the sub-dataset $w|_{\psi}$ and 
(ii) testing whether a property $\phi$ holds for the sub-dataset $w|_{\psi}$.
When no data in the dataset $w$ satisfies the property $\psi$, we can describe this as $\M, w \models \psi \erightarrow \bot$
\modified{by using the propositional constant falsum~$\bot$}.

\modified{Note that the conditioning $\psi \erightarrow \phi$ can be regarded as the modal formula $\MTDt \phi$ with the dataset transformation $T$ where $T(w) = w|_{\psi}$ for all $w\in\calw$.}

\modified{In Sects.~\ref{sec:ML:performance-prediction} and~\ref{sec:ML:robustness}, we show concrete examples using the conditioning operator~$\erightarrow$, 
i.e., the classification performance and robustness of statistical classifiers.}

\subsection{Modality for Conditional Indistinguishability}
\label{sub:modal:indistinguishable}

Next, we introduce a modal operator that is used to formalize the fairness of machine learning in Sect.~\ref{sec:ML:fairness}.

Given two static formulas $\psi_0, \psi_1$ (e.g., representing male and female), $w|_{\psi_0}(x)$ (resp. $w|_{\psi_1}(x)$) represents the probability distribution of values of a measurement variable $x$ generated from the sub-dataset $w|_{\psi_0}$, e.g., the sub-dataset about male (resp. $w|_{\psi_1}$, e.g., about female).
To formalize a certain similarity between $x$'s values generated from the two sub-datasets (e.g., between the benefits for male and for female), we introduce a modal operator $\IND{x}{\varepsilon,D}$ for conditional indistinguishability as follows.
We write $\psi_0 \IND{x}{\varepsilon,D} \psi_1$ to represent that the two distributions $w|_{\psi_0}(x)$ and $w|_{\psi_1}(x)$ are indistinguishable up to a threshold $\varepsilon$ in terms of a divergence or distance $D$.
Formally, this modality is defined as follows.\footnote{The semantics for the (binary) composite operator in the arrow logic~\cite{Blackburn:01:book} resembles that for $\IND{x}{\varepsilon,D}$ in Definition~\ref{def:modal:IND}, although it has a totally different meaning and motivation.}

\begin{definition}[Conditional indistinguishability operator $\IND{x}{\varepsilon,D}$]\label{def:modal:IND} \rm
\!Assume that the universe $\calw$ includes all sub-multisets of each $w\in\calw$.
Given an $x\in\Var$, an $\varepsilon\in\realsnng$, and a divergence or distance $D: \Dists\calo\times\Dists\calo \rightarrow \realsnng$, we define an accessibility relation by:
\begin{align*}
\Rind \eqdef \{ (w_0, w_1)\in\calw\times\calw \,|\,
\diverge{\sigma_{w_0}(x)\!}{\!\sigma_{w_1}(x)} \le \varepsilon \}
{.}
\end{align*}
Then for static formulas $\psi_0$ and $\psi_1$, we define the interpretation of $\psi_0 \IND{x}{\varepsilon,D} \psi_1$ by:
\begin{align*}
&\M, w \models \psi_0 \IND{x}{\varepsilon,D} \psi_1
\\
&\mbox{iff }~
\mbox{ there exist $w_0, w_1$~ s.t. $(w, w_0) \in \RpsiA$, }
\\&\hspace{3.5ex}
\mbox{ $(w, w_1)\in\RpsiB$, and $(w_0, w_1) \in \Rind$}
{,}
\end{align*}
where $\RpsiA$ and $\RpsiB$ are two conditioning relations in Definition~\ref{def:modal:conditioning}.
\end{definition}

Note that two worlds are related by $\Rind$ if they have close probability distributions of the values of $x$.
Intuitively, $w \models \psi_0 \IND{x}{\varepsilon,D} \psi_1$ corresponds to the two operations: 
(i) transforming the given dataset $w$ to the two sub-datasets $w|_{\psi_0}$ and $w|_{\psi_1}$, and
(ii) testing whether the probability distribution of $x$ generated by the dataset $w|_{\psi_0}$ is indistinguishable from \modified{the distribution generated} by the dataset~$w|_{\psi_1}$.

When $\varepsilon = 0$, the operator $\IND{x}{\varepsilon,D}$ represents the identity of two distributions.

\begin{restatable}{prop}{CondIND}
\label{prop:conditional-identity}
For a world $w$, static formulas $\psi_0$, $\psi_1$, and a measurement variable $x$,
$w \models \psi_0 \IND{x}{0,D} \psi_1$
iff
the distribution $w|_{\psi_0}(x)$ is identical to $w|_{\psi_1}(x)$.
\end{restatable}

This proposition is immediate from the following lemma.

\begin{restatable}{lem}{INDdivergence}
\label{lem:IND-divergence}
For a world $w$, static formulas $\psi_0$, $\psi_1$, and a measurement variable $x$,
\begin{align*}
w \models \psi_0 \IND{x}{\varepsilon,D} \psi_1
\mbox{ iff }
\diverge{\sigma_{w|_{\psi_0}}(x)}{\sigma_{w|_{\psi_1}}(x)} \le \varepsilon
{.}
\end{align*}
\end{restatable}

\begin{proof}
Let $w_0 = w|_{\psi_0}$ and $w_1 = w|_{\psi_1}$.
Then by Definition~\ref{def:modal:conditioning}, we have 
$(w, w_0) \in \RpsiA$ and $(w, w_1)\in\RpsiB$.
Hence this lemma follows from Definition~\ref{def:modal:IND}.
\qed
\end{proof}

\modified{In Sect.~\ref{sec:ML:fairness}, we present examples using the conditional indistinguishability operator,
i.e., we formalize various notions of fairness in machine learning by using this operator and the above proposition and lemma.}

\subsection{Summary on the Modal Language}
\label{sub:modal:summary}

In summary, modal operators are used to represent transformation and testing on datasets.
The unary modal operator $\MTDt$ is regarded as a transformation $\trans$ on datasets, 
while the binary modal operators $\erightarrow$ and $\IND{x}{\varepsilon,D}$ are regarded as transforming-then-testing on datasets.

Now the syntax of the formulas is given by:
\begin{align*}
&\mbox{Static formulas:}~
\\[-0.5ex]&
\psi \mathbin{::=}
 \gamma(x_1, x_2, \ldots, x_n) \mid
 \neg \psi \mid \psi \wedge \psi
\\
&\mbox{Dataset formulas:}~
\\[-0.5ex]&
\phi \mathbin{::=}
 \PR{I} \psi \,|\, \neg \phi \,|\, \phi \wedge \phi \,|\,
 \MTDt \phi \,|\, \psi \erightarrow \phi \,|\, \psi_0 \IND{x}{\varepsilon,D} \psi_1 \,|\,
 \MKa \phi,
\end{align*}
where the epistemic formulas with the additional modality are called \emph{dataset formulas}, since they are interpreted in a world that corresponds to a dataset.

When multiple transformations/testing are sequentially applied to datasets, we can use dataset formulas in which different modal operators are nested.
For example, $w \models \MTDt (\psi \erightarrow \phi)$ represents that after applying a data preparation $T$ to a dataset $w$, a property $\phi$ holds for the sub-dataset $T(w)|_{\psi}$ that satisfies $\psi$.

\section{Epistemic Model for Supervised Learning}
\label{sec:formal:ML}

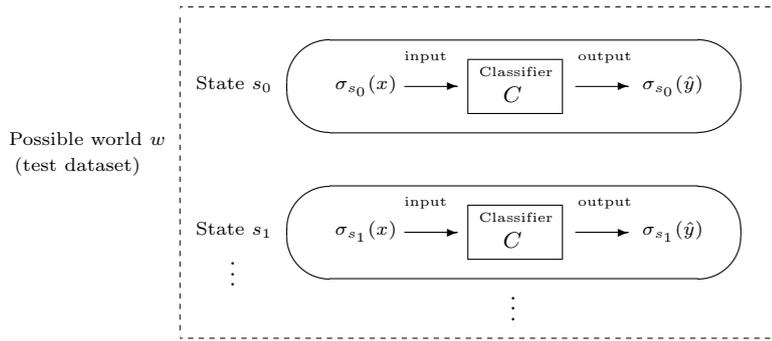
\begin{figure*}[t]
\centering
\begin{picture}(310, 128)
 \put(73,94){\scriptsize State $s_0$}
 \put(192,95){\oval(170,35)}
 \put(175,85){\framebox(35,20)}
 \put(151,105){\tiny input}
 \put(125,94){\scriptsize $\sigma_{s_0}(x)$}
 \put(151,95){\vector(1,0){20}}
 \put(216,105){\tiny output}
 \put(240,94){\scriptsize $\sigma_{s_0}(\hy)$}
 \put(215,95){\vector(1,0){20}}
 \put(179,99){\tiny Classifier}
 \put(188,89){$C$}
 \put(73,39){\scriptsize State $s_1$}
 \put(192,40){\oval(170,35)}
 \put(175,30){\framebox(35,20)}
 \put(151,50){\tiny input}
 \put(125,39){\scriptsize $\sigma_{s_1}(x)$}
 \put(151,40){\vector(1,0){20}}
 \put(216,50){\tiny output}
 \put(240,39){\scriptsize $\sigma_{s_1}(\hy)$}
 \put(215,40){\vector(1,0){20}}
 \put(179,44){\tiny Classifier}
 \put(188,34){$C$}
 \put(85,19){\rotatebox{90}{$\cdots$}}
 \put(190,6){\rotatebox{90}{\footnotesize $\cdots$}}
 \put(3,73){\scriptsize Possible world $w$}
 \put(5,63){\scriptsize (test dataset)}
 \linethickness{0.1pt}
 \put(67,0){\dashbox{2}(225,125)}
\end{picture}
\caption{
A world $w$ is chosen non-deterministically and corresponds to a test dataset.
With probability $w[s_i]$, the world $w$ is in a deterministic state $s_i$ where the classifier $C$ receives the input value $\sigma_{s_i}(x)$ and returns the output value $\sigma_{s_i}(\hy)$.
Each state $s_i$ can be regarded as a tuple $(\sigma_{s_i}(x), \sigma_{s_i}(y), \sigma_{s_i}(\hy)) \in\cald\times\Label\times\Label$ consisting of an input datum, an actual label, and a predicted label.
\label{fig:states}}
\end{figure*}

In this section we introduce a formal model for supervised learning.
Specifically, we employ a distributional Kripke model (Definition~\ref{def:dist-Kripke-model}), and formalize a behavior of a classifier $C$ and a non-deterministic input $x$ from an adversary in the model.
In this formalization, we focus only on the testing of supervised learning models, and do \emph{not} formalize the training of supervised learning models or learning algorithms themselves.

\subsection{Classification Problems}
\label{sub:ML:classifications}

\emph{Multiclass classification} is the problem of classifying a given input into one of multiple classes.
Let $\Label$ be a finite set of \emph{class labels}%
\footnote{The regression can be regarded as the classification problem when the label ranges over the real numbers, hence it can be formalized using a distributional Kripke model analogously. For simplicity, however, we deal only with the classification problems in this paper.}, 
and $\cald$ be a finite set of \emph{input data} (called \emph{feature vectors}) that we want to classify.
Then a \emph{classifier} is a function $C: \cald\rightarrow\Label$ that receives an input  datum $v$ and predicts which class (among $\Label$) the input $v$ belongs to.
In this work, we deal with a situation where some classifier $C$ has already been obtained and its properties should be evaluated, and do \emph{not} model or reason about how classifiers are trained from a training dataset.

We assume a \emph{scoring function} $f: \cald\times\Label \rightarrow \reals$ that gives a score $f(v, \ell)$ of predicting the class of an input datum (feature vector) $v$ as a label $\ell$.
Then for each input $v\in\cald$, we denote by $H(v) = \ell$ to represent that a label $\ell$ maximizes $f(v, \ell)$.
For example, when the input $v$ is an image of an animal and $\ell$ is the animal's name, then $H(v) = \ell$ may represent that an oracle (or a ``human'') classifies the image $v$ as $\ell$.

\subsection{Modeling the Behaviors of Classifiers}
\label{sub:ML:predicates}

A classifier is formalized on a distributional Kripke \allowbreak model 
$\M =(\calw, (\calr_a)_{a\in\cala}, \allowbreak (V_s)_{s\in\cals})$ 
with $\calw = \Dists\cals$.
Then $\calw$ is an infinite set of possible worlds that corresponds to 
\emph{all possible datasets} \modified{from which the classifier can receive input data}.
We denote by $\wtest \in \calw$ a real world that corresponds to a test dataset.
Recall that each world $w\in\calw$ is a multiset of states over $\cals$ and is associated with a stochastic assignment $\sigma_w: \Var \rightarrow \Dists\calo$ that is consistent with the deterministic assignments~$\sigma_s$ for all $s\in w$, as explained in Sect.~\ref{sub:Kripke}.

We present an overview of our formalization in Fig.~\ref{fig:states}.
We denote by $x\in\Var$ an input datum given to the classifier $C$ (and to the oracle $H$), by $y\in\Var$ a correct label given by the oracle $H$, and by $\hy\in\Var$ a label predicted by $C$.
We assume that the input variable $x$ (resp. the output variables $y,\hy$) ranges over the set $\cald$ of input data (resp. the set $\Label$ of labels);
i.e., the deterministic assignment $\sigma_s$ at each state $s\in\cals$ has the range $\calo = \cald \cup \Label$ and satisfies $\sigma_s(x)\in\cald$ and $\sigma_s(y), \sigma_s(\hy)\in\Label$.

A key idea in our modeling is that we describe logical aspects of statistical properties in the syntax level by using logical formulas, and model statistical distances and dataset operations in the semantics level by using accessibility relations in the distributional Kripke model.
In this way, we can formalize various statistical properties of classifiers in a simple and abstract way.

To formalize the classifier $C$, we introduce a static formula $\psi(x, \hy)$ to represent that $C$ classifies a given input $x$ as a class $\hy$.
We also introduce a static formula $h(x, y)$ to represent that $y$ is the actual class of an input $x$.
As an abbreviation, we write $\psi_\ell(x)$ (resp. $h_\ell(x)$) to denote $\psi(x, \ell)$ (resp. $h(x, \ell)$).
Formally, these static formulas are interpreted at each state $s\in\cals$ as follows:
\begin{align*}
s \models \psi(x, \hy) &~\mbox{ iff }~
C(\sigma_s(x)) = \sigma_s(\hy).
\\
s \models h(x, y) &~\mbox{ iff }~
H(\sigma_s(x)) = \sigma_s(y).
\end{align*}

\subsection{Modeling the Non-deterministic Inputs from Adversaries}
\label{sub:ML:non-deterministic-inputs}

We first observe that a distributional Kripke model $\M$ can formalize an input $x$ that is probabilistically chosen from a given dataset.
As explained in Sect.~\ref{sub:Kripke},
each world $w$ corresponds to a \emph{test dataset}.
When a state $s$ is drawn from a multiset $w$ of states, an input value $\sigma_s(x)$ is sampled from the distribution $\sigma_w(x)$, and assigned to the measurement variable~$x$.
The set of all possible probability distributions of inputs is represented by
$\Lambda \eqdef \left\{ \sigma_w(x) \mid w\in\calw \right\}$,
which is possibly an infinite set.

For example, let us consider testing the classifier $C$ with the actual test dataset $\sigma_{\wtest}(x)$.
When $C$ classifies an input $x$ as a label $\ell$ with probability $0.2$, i.e.,
\[
\Pr\!\left[~ v \randassign \sigma_{\wtest}(x) \,:\, 
C(v) = \ell ~\right] = 0.2
,\]
then this can be expressed by:
\begin{align*}
\M, \wtest \models \PR{0.2} \psi_\ell(x)
{.}
\end{align*}

Next we observe that our model can formalize a non-deterministic input $x$ from an adversary as follows.
Although each state $s$ in a possible world $w$ is assigned the probability $w[s]$, each world $w$ itself is not assigned a probability.
Thus, each input distribution $\sigma_w(x) \in \Lambda$ itself is also not assigned a probability, hence our model assumes no probability distribution over $\Lambda$.
In other words, we assume that a world $w$ and thus an input distribution $\sigma_w(x)$ are non-deterministically chosen.
This is useful to model an adversary that provides malicious inputs to the classifier $C$ to make its prediction fail,
because we usually do not have a prior knowledge of the probability distribution of malicious inputs from adversaries, and need to reason about the worst cases caused by the attack.
In Sect.~\ref{sec:ML:robustness}, this formalization of non-deterministic inputs is used to express the robustness of classifiers.

Finally, it should be noted that we cannot enumerate all possible adversarial inputs, hence cannot enumerate all possible datasets to construct the universe $\calw$.
Since $\calw$ can be an infinite set and is unspecified, we cannot check whether a formula expressing a security property against an adversary is satisfied in all possible worlds of $\calw$.
Nevertheless, as shown in later sections, describing various properties using our extension of \StatEL{} is useful to explore desirable properties and to discuss relationships among them.

\section{Formalizing the Classification Performance}
\label{sec:ML:performance-prediction}
In this section we show a formalization of classification performance using our extension of \StatEL{}.
We formalize popular measures of classification performance, including precision, recall, and accuracy,
and measures for evaluating overfitting, such as the generalization error.
See Fig.~\ref{fig:performance-generalization} for basic ideas on these formalizations.

\begin{table*}[t]
  \caption{Logical description of the table of confusion \label{table:confusion}}
  \centering
  \scalebox{0.9}{
  \begin{tabular}{|l|l|l|l|l|} \cline{1-3}
    \multicolumn{1}{|c|}{}& \multicolumn{2}{c|}{Actual class} \\ \cline{2-5}
    \multicolumn{1}{|c|}{}&
    \multicolumn{1}{c|}{positive} & \multicolumn{1}{c|}{negative}
    &\!$\Prevalence_{\ell,I}(x)$
    &\!$\Accuracy_{\ell,I}(x)$ \\[-0.3ex]
    \multicolumn{1}{|c|}{}&
    \multicolumn{1}{c|}{$h_\ell(x)$} & \multicolumn{1}{c|}{$\neg h_\ell(x)$}
    &\!$\eqdef\!\PR{I}(h_\ell(x))$
    &\!$\eqdef\!\PR{I}(\psi_\ell(x) \leftrightarrow h_\ell(x))$\! \\[0.5ex] \hline
    \!Positive & & & & \\[-1.1ex]
    \!prediction\! & ~$\TP(x) \eqdef$ & ~$\FP(x) \eqdef$ 
    & $\Precision_{\ell,I}(x) \eqdef$
    & $\FDR_{\ell,I}(x) \eqdef$ \\
    $\psi_\ell(x)$ &
    ~$\psi_\ell(x) \wedge h_\ell(x)$ & ~$\psi_\ell(x) \wedge \neg h_\ell(x)$
    & $\psi_\ell(x) \erightarrow \PR{I} h_\ell(x)$
    & $\psi_\ell(x) \erightarrow \PR{I} \neg h_\ell(x)$ \\[0.5ex] \cline{1-5}
    \!Negative & & & & \\[-1.1ex]
    \!prediction\! & ~$\FN(x) \eqdef$ & ~$\TN(x) \eqdef$ 
    & $\FOR_{\ell,I}(x) \eqdef$
    & $\NPV_{\ell,I}(x) \eqdef$ \\
    $\neg\psi_\ell(x)$ &
    ~$\neg \psi_\ell(x) \wedge h_\ell(x)$ & ~$\neg \psi_\ell(x) \wedge \neg h_\ell(x)$
    & $\neg \psi_\ell(x) \erightarrow \PR{I} h_\ell(x)$\!
    & $\neg \psi_\ell(x) \erightarrow \PR{I} \neg h_\ell(x)$ \\[0.5ex] \hline
    \multicolumn{1}{c|}{} & \!$\Recall_{\ell,I}(x) \eqdef$~ 
    & \!$\FallOut_{\ell,I}(x) \eqdef$~ \\
    \multicolumn{1}{c|}{} & $h_\ell(x) \erightarrow \PR{I} \psi_\ell(x)$~~
    & \!$\neg h_\ell(x) \erightarrow \PR{I} \psi_\ell(x)$~~ \\[0.5ex] \cline{2-3}
    \multicolumn{1}{c|}{} & \!$\MissRate_{\ell,I}(x) \eqdef$~ 
    & \!$\Specificity_{\ell,I}(x) \eqdef$~ \\
    \multicolumn{1}{c|}{} & $h_\ell(x) \erightarrow \PR{I} \neg\psi_\ell(x)$
    & \!$\neg h_\ell(x) \erightarrow \PR{I} \neg\psi_\ell(x)$ \\[0.5ex] \cline{2-3}
  \end{tabular}
}
\end{table*}

\subsection{Classifier's Prediction and its Correctness}
\label{sub:ML:predict-correct}

In classification problems, the terms \emph{positive}/\emph{negative} represent the result of the classifier's prediction, and the terms \emph{true}/\emph{false} represent whether the classifier predicts correctly or not.
Then the following terminologies are commonly used:
\begin{itemize}
\item \emph{true positive} ($\TP$): both the prediction and actual class are positive;
\item \emph{true negative} ($\TN$): both the prediction and actual class are negative;
\item \emph{false positive} ($\FP$): the prediction is positive but the actual class is negative;
\item \emph{false negative} ($\FN$): the prediction is negative but the actual class is positive.
\end{itemize}
These terminologies can be formalized using static formulas as shown in Table~\ref{table:confusion}.
For example, when an input $x$ shows true positive at a state $s$, this can be expressed as $s \models \psi_\ell(x) \wedge h_\ell(x)$.
Note that the value of the measurement variable $x$ is uniquely determined by the assignment $\sigma_s$ at the state $s$.
True negative, false positive (type I error), and false negative (type II error) are respectively expressed as 
$s \models \neg \psi_\ell(x) \wedge \neg h_\ell(x)$,\, 
$s \models \psi_\ell(x) \wedge \neg h_\ell(x)$, and 
$s \models \neg \psi_\ell(x) \wedge h_\ell(x)$.

\begin{figure*}[t]
\centering
\begin{tikzpicture}
\coordinate (W0) at (-2.0,2.9) node at (W0) [right] {{\scriptsize Real world $\wtest$}};
\coordinate (W0) at (-2.0,2.6) node at (W0) [right] {{\scriptsize with a test dataset}};
\draw [black, dotted, name path=rectangle, rotate=0] (1.0,2.0) rectangle +(8.8,2.75);
\coordinate (W1) at (-2.0,0.9) node at (W1) [right] {{\scriptsize Possible world $\wtrain$}};
\coordinate (W0) at (-2.0,0.6) node at (W0) [right] {{\scriptsize with a training dataset}};
\draw [black, dotted, name path=rectangle, rotate=0] (1.0,-1.2) rectangle +(8.8,2.75);
\coordinate (Dataset) at (1.85,3.8) node at (Dataset) {{\scriptsize Distribution}};
\coordinate (Dataset) at (1.85,3.5) node at (Dataset) {{\scriptsize of test data}};
\filldraw [gray!15] (1.8,2.5) circle [x radius=6mm, y radius=2mm, rotate=0];
\draw [black, very thin] (1.8,2.5) circle [x radius=6mm, y radius=2mm, rotate=0];
\fill [gray!15!white] (1.2,2.5) rectangle (2.4,3);
\filldraw [gray!15] (1.8,3) circle [x radius=6mm, y radius=2mm, rotate=0];
\draw [black, very thin] (1.8,3) circle [x radius=6mm, y radius=2mm, rotate=0];
\draw [black, very thin] (1.2,2.5) -- (1.2,3.0);
\draw [black, very thin] (2.4,2.5) -- (2.4,3.0);
\coordinate (DB) at (1.83,2.6) node at (DB) {$\sigma_{\!\wtest}(x)$};
\coordinate (Dataset) at (2.0,0.05) node at (Dataset) {{\scriptsize Distribution of}};
\coordinate (Dataset) at (2.0,-0.25) node at (Dataset) {{\scriptsize training data}};
\filldraw [gray!15] (1.8,0.5) circle [x radius=6mm, y radius=2mm, rotate=0];
\draw [black, very thin] (1.8,0.5) circle [x radius=6mm, y radius=2mm, rotate=0];
\fill [gray!15!white] (1.2,0.5) rectangle (2.4,1);
\filldraw [gray!15] (1.8,1) circle [x radius=6mm, y radius=2mm, rotate=0];
\draw [black, very thin] (1.8,1) circle [x radius=6mm, y radius=2mm, rotate=0];
\draw [black, very thin] (1.2,0.5) -- (1.2,1.0);
\draw [black, very thin] (2.4,0.5) -- (2.4,1.0);
\coordinate (DB2) at (1.83,0.6) node at (DB2) {$\sigma_{\!\wtrain}(x)$};
\filldraw [gray!5!white, name path=rectangle, rotate=0] (5.2,3.6) rectangle +(1.3,1);
\draw [name path=rectangle, rotate=0] (5.2,3.6) rectangle +(1.3,1);
\coordinate (HI1) at (4.2,3.1) node at (HI1) {};
\coordinate (HI2) at (5.1,4.1) node at (HI2) {};
\coordinate (HL1) at (5.8,4.25) node at (HL1) [above] {{\scriptsize Oracle}};
\coordinate (HL2) at (5.8,3.97) node at (HL2) [above] {{\tiny (human)}};
\coordinate (H) at (6.0,3.87) node at (H) [left] {$H$};
\coordinate (Hout) at (6.6,4.1) node at (Hout) {};
\coordinate (HO) at (6.95,4.1) node at (HO) [right] {$\sigma_{\!s}(y)$};
\draw [->] (HI1)--(HI2);
\draw [->] (Hout)--(HO);
\coordinate (INPUT1) at (4.45,3.75) node at (INPUT1) [above] {{\tiny input}};
\coordinate (OUTPUT1) at (7.1,4.3) node at (OUTPUT1) [above] {{\tiny output}};
\filldraw [gray!5!white, name path=rectangle, rotate=0] (5.2,2.3) rectangle +(1.3,1);
\draw [name path=rectangle, rotate=0] (5.2,2.3) rectangle +(1.3,1);
\coordinate (SL) at (2.55,2.8) node at (SL) {};
\coordinate (S) at (2.7,3.15) node at (S) {~~~~~~{\tiny sampling}};
\coordinate (SR) at (3.45,2.8) node at (SR) {};
\coordinate (I) at (4.3,2.8) node at (I) [left] {$\sigma_{\!s}(x)$};
\coordinate (CI) at (5.1,2.8) node at (CI) {};
\coordinate (L) at (5.85,2.85) node at (L) [above] {{\scriptsize Classifier}};
\coordinate (C) at (6.0,2.65) node at (C) [left] {$C$};
\coordinate (CO) at (6.6,2.8) node at (CO) {};
\coordinate (O) at (6.95,2.8) node at (O) [right] {$\sigma_{\!s}(\hy)$};
\draw [->, densely dashed] (SL)--(SR);
\draw [->] (I)--(CI);
\draw [->] (CO)--(O);
\filldraw [gray!5!white, name path=rectangle, rotate=0] (5.2,0.3) rectangle +(1.3,1);
\draw [name path=rectangle, rotate=0] (5.2,0.3) rectangle +(1.3,1);
\coordinate (SL2) at (2.55,0.8) node at (SL2) {};
\coordinate (S2) at (2.7,1.15) node at (S2) {~~~~~~{\tiny sampling}};
\coordinate (SR2) at (3.45,0.8) node at (SR2) {};
\coordinate (I2) at (4.4,0.8) node at (I2) [left] {$\sigma_{\!s'}(x)$};
\coordinate (CI2) at (5.1,0.8) node at (CI2) {};
\coordinate (L2) at (5.85,0.85) node at (L2) [above] {{\scriptsize Classifier}};
\coordinate (C2) at (6.0,0.65) node at (C2) [left] {$C$};
\coordinate (CO2) at (6.6,0.8) node at (CO2) {};
\coordinate (O2) at (6.95,0.8) node at (O2) [right] {$\sigma_{\!s'}(\hy)$};
\draw [->, densely dashed] (SL2)--(SR2);
\draw [->] (I2)--(CI2);
\draw [->] (CO2)--(O2);
\filldraw [gray!5!white, name path=rectangle, rotate=0] (5.2,-1.0) rectangle +(1.3,1);
\draw [name path=rectangle, rotate=0] (5.2,-1.0) rectangle +(1.3,1);
\coordinate (H2I1) at (4.2,0.5) node at (H2I1) {};
\coordinate (H2I2) at (5.1,-0.5) node at (H2I2) {};
\coordinate (H2L1) at (5.8,-0.35) node at (H2L1) [above] {{\scriptsize Oracle}};
\coordinate (H2L2) at (5.8,-0.63) node at (H2L2) [above] {{\tiny (human)}};
\coordinate (H2) at (6.0,-0.73) node at (H2) [left] {$H$};
\coordinate (H2out) at (6.6,-0.5) node at (H2out) {};
\coordinate (H2O) at (6.95,-0.5) node at (H2O) [right] {$\sigma_{\!s'}(y)$};
\draw [->] (H2I1)--(H2I2);
\draw [->] (H2out)--(H2O);
\coordinate (INPUT2) at (4.45,-0.85) node at (INPUT2) [above] {{\tiny input}};
\coordinate (OUTPUT2) at (7.1,-1.05) node at (OUTPUT2) [above] {{\tiny output}};
\coordinate (CompD0) at (1.9,2.15) node at (CompD0) {};
\coordinate (CompD1) at (1.9,1.35) node at (CompD1) {};
\coordinate (CompHC0) at (7.15,3.85) node at (CompHC0) {};
\coordinate (CompHC1) at (7.15,3.1) node at (CompHC1) {};
\coordinate (Performance1) at (8.1,3.7) node at (Performance1) {{\scriptsize Performance}};
\coordinate (Loss1) at (8.55,3.3) node at (Loss1) {{\scriptsize Generalization error}};
\draw [<->,double] (CompHC0)--(CompHC1);
\coordinate (General12) at (8.8,1.78) node at (General12) {{\scriptsize Overfitting}};
\coordinate (CompL0) at (8.05,3.0) node at (CompL0) {};
\coordinate (CompL1) at (8.05,0.35) node at (CompL1) {};
\draw [<->,double] (CompL0)--(CompL1);
\coordinate (CompHC0) at (7.15,0.5) node at (CompHC0) {};
\coordinate (CompHC1) at (7.15,-0.25) node at (CompHC1) {};
\coordinate (Loss2) at (8.2,0.1) node at (Loss2) {{\scriptsize Training error}};
\draw [<->,double] (CompHC0)--(CompHC1);
\end{tikzpicture}
\caption{
The classification performance compares the oracle $H$'s output with that of the classifier $C$'s,
while the evaluation of overfitting compares the expected loss by the test dataset with that by the training dataset.
\label{fig:performance-generalization}}
\end{figure*}
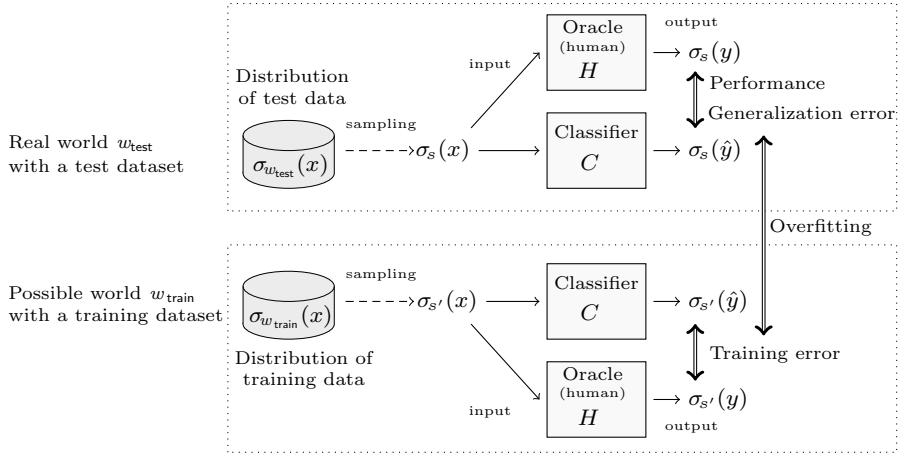

\subsection{Precision, Recall, Accuracy, and Other Performance Measures}
\label{sub:ML:correctness:measures}

Next we formalize three popular measures for binary classification performance: \emph{precision}, \emph{recall}, and \emph{accuracy}.
In Table~\ref{table:confusion} we summarize the formalization of various notions of classification performance using our dataset formulas.

In theory, these notions should be formalized with the infinite dataset $\wtrue$ representing the true distribution.
However, we usually cannot obtain $\wtrue$ or test the performance measures using $\wtrue$.
Hence, we often sample a finite test dataset $\wtest$ from the true distribution and regard it as an approximation of $\wtrue$.%
\footnote{Since the test dataset $\wtest$ is finite, there can be \emph{missing data} that are not included in $\wtest$ but are sampled from the true distribution $\wtrue$ with a very small probability.}

Given a test dataset $\wtest$, \emph{precision} (\emph{positive predictive value}) is defined as the conditional probability that the prediction is correct given that the prediction is positive; i.e., 
${\it precision} = \frac{\TP}{\TP + \FP}$.
Since the probability distribution of the input $x$ in the world $\wtest$ is expressed by $\sigma_{\wtest}(x)$ as explained in Sect.~\ref{sub:ML:non-deterministic-inputs},
the precision being within an interval $I$ is given by:
\begin{align*}
\Pr\!\left[~ v \randassign \sigma_{\wtest}(x) \,:\, 
H(v) = \ell ~\Big|~ C(v) = \ell ~\right] \in I
{,}
\end{align*}
which can be written as:
\begin{align*}
\Pr\!\left[~ s \randassign \wtest \,:\, 
s \models h_\ell(x) ~\Big|~ s \models \psi_\ell(x) ~\right] \in I
{.}
\end{align*}
By using \StatEL{}, this can be formalized as:
\begin{align*}
&\M, \wtest \models \Precision_{\ell,I}(x)
\\&
\mbox{where }~
\Precision_{\ell,I}(x) \eqdef \psi_\ell(x) \erightarrow \PR{I} h_\ell(x)
{.}
\end{align*}
Here $\erightarrow$ is the conditioning operator defined in Sect.~\ref{sub:modal:conditioning}.
The value of precision depends on the test dataset $\wtest$, and can be computed in finite time since $\wtest$ is finite.

Symmetrically, \emph{recall} (\emph{true positive rate}) is defined as the conditional probability that the prediction is correct given that the actual class is positive; i.e., 
${\it recall} = \frac{\TP}{\TP + \FN}$.
Then the recall being within $I$ is formalized as:
\begin{align*}
\Recall_{\ell,I}(x) \eqdef h_\ell(x) \erightarrow \PR{I} \psi_\ell(x)
{.}
\end{align*}

Finally, \emph{accuracy} is the probability that the classifier predicts correctly; i.e., 
${\it accuracy} = \frac{\TP + \TN}{\TP + \TN + \FP + \FN}$.
Then the accuracy being within $I$ is formalized~as:
\begin{align*}
\Accuracy_{\ell,I}(x) \eqdef
\PR{I}\bigl( \psi_\ell(x) \leftrightarrow h_\ell(x) \bigr)
{,}
\end{align*}
which can also be defined as 
$\PR{I}\bigl( \TP(x) \vee \TN(x) \bigr)$.
When we measure the accuracy after a data preparation operation $T$ (e.g., data cleaning) to the test dataset $\wtest$, this can be represented by
$\wtest \models \MTDt \Accuracy_{\ell,I}(x)$.

\begin{example}[Performance of pedestrian detection]\label{eg:human:performance}
\modified{Let us consider an autonomous car that uses a machine learning classifier
to detect a person crossing the road.
For the sake of simplicity, we formalize an example of a binary classifier $C$ that detects whether or not a pedestrian is crossing the road in a photo image in a test dataset $\wtest$.
We write $\sunny(x)$ (resp. $\snowy(x)$) to represent that a photo $x$ was taken on a sunny (resp. snowy) day.
Let $\psi_\ell(x)$ (resp. $h_\ell(x)$) represent that the classifier $C$ (resp. the human) detects a pedestrian crossing the road in an image~$x$.}

\modified{%
We empirically measure recall 
(i.e., the conditional probability that $C$ detects a pedestrian crossing the road when the input image $x$ actually includes it) by using the data 
collected on sunny days.
When $C$ achieves a recall of $0.95$ on sunny days, 
this is represented by
$\wtest \models \sunny(x) \erightarrow \Recall_{\ell,0.95}(x)$.}%

\modified{Since $C$ should detect a pedestrian also on a snow-covered road, it should be tested with the data collected on snowy days.
If we have a recall of $0.8$ 
on snowy days,
this is represented by
$\wtest \models \snowy(x) \erightarrow \Recall_{\ell,0.8}(x)$.}

\modified{More generally, if the classifier $C$ achieves a recall of more than $0.9$ in situations $\gamma_1, \gamma_2, \ldots, \gamma_m$, this can be represented by
$\wtest \models \bigwedge_{i=1}^{m} \bigl( \gamma_i(x) \erightarrow \Recall_{\ell,(0.9, 1]}(x) \bigr)$.
}
\end{example}

\begin{figure*}[t]
\centering
\begin{tikzpicture}
\coordinate (W0) at (-1.7,2.8) node at (W0) [right] {{\scriptsize Real world $\wtest$}};
\draw [black, dotted, name path=rectangle, rotate=0] (0.9,2.0) rectangle +(7.45,1.95);
\coordinate (W1) at (-1.7,0.8) node at (W1) [right] {{\scriptsize Possible world $w'$}};
\draw [black, dotted, name path=rectangle, rotate=0] (0.9,-0.4) rectangle +(7.45,1.95);
\coordinate (Dataset) at (1.75,3.8) node at (Dataset) {{\scriptsize Distribution}};
\coordinate (Dataset) at (1.75,3.5) node at (Dataset) {{\scriptsize of test data}};
\filldraw [gray!15] (1.8,2.5) circle [x radius=6mm, y radius=2mm, rotate=0];
\draw [black, very thin] (1.8,2.5) circle [x radius=6mm, y radius=2mm, rotate=0];
\fill [gray!15!white] (1.2,2.5) rectangle (2.4,3);
\filldraw [gray!15] (1.8,3) circle [x radius=6mm, y radius=2mm, rotate=0];
\draw [black, very thin] (1.8,3) circle [x radius=6mm, y radius=2mm, rotate=0];
\draw [black, very thin] (1.2,2.5) -- (1.2,3.0);
\draw [black, very thin] (2.4,2.5) -- (2.4,3.0);
\coordinate (DB) at (1.83,2.6) node at (DB) {$\sigma_{\!\wtest}(x)$};
\coordinate (Dataset) at (1.9,0.05) node at (Dataset) {{\scriptsize Distribution of}};
\coordinate (Dataset) at (1.9,-0.25) node at (Dataset) {{\scriptsize perturbed data}};
\filldraw [gray!15] (1.8,0.5) circle [x radius=6mm, y radius=2mm, rotate=0];
\draw [black, very thin] (1.8,0.5) circle [x radius=6mm, y radius=2mm, rotate=0];
\fill [gray!15!white] (1.2,0.5) rectangle (2.4,1);
\filldraw [gray!15] (1.8,1) circle [x radius=6mm, y radius=2mm, rotate=0];
\draw [black, very thin] (1.8,1) circle [x radius=6mm, y radius=2mm, rotate=0];
\draw [black, very thin] (1.2,0.5) -- (1.2,1.0);
\draw [black, very thin] (2.4,0.5) -- (2.4,1.0);
\coordinate (DB2) at (1.83,0.6) node at (DB2) {$\sigma_{\!w'}(x)$};
\filldraw [gray!5!white, name path=rectangle, rotate=0] (5.2,2.3) rectangle +(1.4,1);
\draw [name path=rectangle, rotate=0] (5.2,2.3) rectangle +(1.4,1);
\coordinate (SL) at (2.55,2.8) node at (SL) {};
\coordinate (S) at (2.7,3.15) node at (S) {~~~~~~{\tiny sampling}};
\coordinate (SR) at (3.45,2.8) node at (SR) {};
\coordinate (I) at (4.5,2.8) node at (I) [left] {$\sigma_{\!s}(x)$};
\coordinate (CI) at (5.1,2.8) node at (CI) {};
\coordinate (L) at (5.9,2.85) node at (L) [above] {{\scriptsize Classifier}};
\coordinate (C) at (6.1,2.65) node at (C) [left] {$C$};
\coordinate (CO) at (6.8,2.8) node at (CO) {};
\coordinate (O) at (7.3,2.8) node at (O) [right] {~$\sigma_{\!s}(\hy)$};
\draw [->, densely dashed] (SL)--(SR);
\draw [->] (I)--(CI);
\draw [->] (CO)--(O);
\coordinate (INPUT1) at (4.7,3.2) node at (INPUT1) [above] {{\tiny input}};
\coordinate (OUTPUT1) at (7.15,3.2) node at (OUTPUT1) [above] {{\tiny output}};
\filldraw [gray!5!white, name path=rectangle, rotate=0] (5.2,0.3) rectangle +(1.4,1);
\draw [name path=rectangle, rotate=0] (5.2,0.3) rectangle +(1.4,1);
\coordinate (SL2) at (2.55,0.8) node at (SL2) {};
\coordinate (S2) at (2.7,1.15) node at (S2) {~~~~~~{\tiny sampling}};
\coordinate (SR2) at (3.45,0.8) node at (SR2) {};
\coordinate (I2) at (4.55,0.8) node at (I2) [left] {$\sigma_{\!s'}(x)$};
\coordinate (CI2) at (5.1,0.8) node at (CI2) {};
\coordinate (L2) at (5.9,0.85) node at (L2) [above] {{\scriptsize Classifier}};
\coordinate (C2) at (6.1,0.65) node at (C2) [left] {$C$};
\coordinate (CO2) at (6.8,0.8) node at (CO2) {};
\coordinate (O2) at (7.3,0.8) node at (O2) [right] {~$\sigma_{\!s'}(\hy)$};
\draw [->, densely dashed] (SL2)--(SR2);
\draw [->] (I2)--(CI2);
\draw [->] (CO2)--(O2);
\coordinate (CompD0) at (1.8,2.15) node at (CompD0) {};
\coordinate (CompD1) at (1.8,1.35) node at (CompD1) {};
\draw [<->,double] (CompD0)--(CompD1);
\coordinate (CompL0) at (7.7,2.45) node at (CompL0) {};
\coordinate (CompL1) at (7.7,1.1) node at (CompL1) {};
\coordinate (RelationL) at (1.3,1.76) node at (RelationL) {$\Rrob$};
\coordinate (Robust) at (8.55,1.78) node at (Robust) {{\scriptsize Robustness}};
\draw [<->,double] (CompL0)--(CompL1);
\coordinate (CompHC0) at (7.7,4.0) node at (CompHC0) {};
\coordinate (CompHC1) at (7.7,3.15) node at (CompHC1) {};
\end{tikzpicture}
\caption{
The robustness compares the conditional probability in the test dataset $\wtest$ with that in another possible world $w'$ that is close to $\wtest$ in terms of $\Rrob$.
Note that an adversary's choice of the input distribution $\sigma_{w'}(x)$ is formalized as a non-deterministic choice of the possible world $w'$.
\label{fig:performance-robustness}}
\end{figure*}
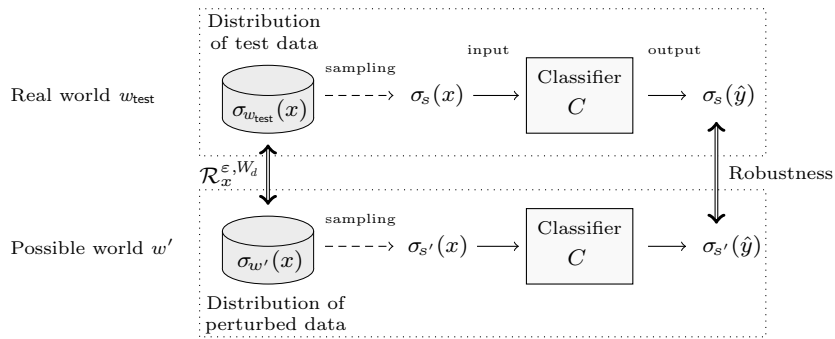

\subsection{Generalization Error}
\label{sub:ML:correctness:generalize}

We next formalize the \emph{generalization error} of a classifier, i.e., a measure of how accurately a classifier is able to predict the class of \emph{previously unseen} input data.
Since a classifier has been trained on a finite sample training dataset $\wtrain$, it may be \emph{overfitted} to $\wtrain$ and have worse classification performance on new input data that have not been included in $\wtrain$.

To formalize the generalization error, we introduce a formula $\Loss(y, \hy)$ to represent that given a correct label $y$ and a predicted label $\hy$, the expected value of losses (i.e., real numbers representing the penalty for incorrect classification) is at most
a non-negative real number $L$.
Formally, the semantics of $\Loss(y, \hy)$ is given by:
\begin{align*}
w \models \Loss(y, \hy)
~~\mbox{ iff }~
\expect_{(v, \hv) \sim \sigma_{w}(y, \hy) }\hspace{-2ex}
\loss(v, \hv) \le L
{,}
\end{align*}
where $\loss$ is a loss function selected according to the data domain $\calo$, and a pair $(v, v')$ of a correct label and a predicted label follows the joint distribution $\sigma_{w}(y, \hy)$.

Now the generalization error being $L$ or smaller at a true distribution $\wtrue$ is written as $\wtrue \models \GE_{L}(x, y, \hy)$ where:
\begin{align*}
\GE_{L}(x, y, \hy) \eqdef
\bigl( h(x, y) \wedge \psi(x, \hy) \bigr) \erightarrow \Loss(y, \hy)
{.}
\end{align*}

Since we usually cannot obtain the true distribution $\wtrue$ and cannot check the satisfaction $\wtrue \models \GE_{L}(x, y, \hy)$, we often compute an empirical error (as an approximation of the generalization error) by using a finite test dataset $\wtest$ that is believed to be an approximation of $\wtrue$.
This testing can be expressed as $\wtest \models \GE_{L}(x, y, \hy)$.

On the other hand, given a training dataset $\wtrain$, the \emph{training error} being at most $\Ltrain$ is represented by $\wtrain \models \GE_{\Ltrain}(x, y, \hy)$.
Then the overfitting of the classifier can be evaluated by comparing the empirical error $L$ with the training error $\Ltrain$.
When the empirical error is smaller than $\Ltrain + \varepsilon$ for some error bound $\varepsilon > 0$, this can be represented by $\wtest \models \GE_{\Ltrain + \varepsilon}(x, y, \hy)$.

\section{Formalizing the Robustness of Classifiers}
\label{sec:ML:robustness}
Many recent studies have found attacks on machine learning where a malicious adversary manipulates the input to cause a malfunction in a machine learning task~\cite{Chakraborty:18:arxiv}.
Such input data, called \emph{adversarial examples}~\cite{Szegedy:14:ICLR}, are designed to make a classifier fail to predict the actual class $\ell$ of the input, but are recognized to belong to $\ell$ from human eyes.
In computer vision, for example, Goodfellow et al.~\cite{Goodfellow:ICLR:15} create an adversarial example by adding undetectable noise to a panda's photo so that humans can still recognize the perturbed image as a panda, but a classifier misclassifies it as a gibbon.
To prevent or mitigate such attacks, the classifier should be \emph{robust} against perturbed input, i.e., 
it should return similar predicted labels given similar input data.

In this section we formalize robustness notions for classifiers by using epistemic operators in \StatEL{} (See Fig.~\ref{fig:performance-robustness} for an overview of the formalization).
Furthermore, we show certain relationships between classification performance and robustness, and suggest a class of robustness properties that have not been formalized in the literature as far as we know.
\modified{We present an overview of these formalizations and relationships in Fig.~\ref{fig:compare-robustness}.}

\subsection{Total Correctness of Classifiers}
\label{sub:security-classifiers}

We first note that the \emph{total correctness} of classifiers could be \modified{formalized} as a classification performance (e.g., precision, recall, or accuracy) in the presence of all possible inputs from adversaries.
For example, the total correctness could be formalized as $\M \models \Recall_{\ell,I}(x)$, which represents that $\Recall_{\ell,I}(x)$ is 
\modified{satisfied} in all possible worlds of $\M$.

In practice, however, it is not possible or tractable to test whether the classification performance is achieved for all possible test datasets
(corresponding to an infinite number of possible worlds in $\M$).
Hence we need a weaker form of a correctness notion, which may be verified or tested in some way.
In the following sections, we deal with robustness notions that are weaker than total correctness.

\subsection{Accessibility Relation for Robustness}
\label{sub:relation-robustness}

To formalize robustness notions, we introduce an accessibility relation $\Rrob$ that relates two worlds having closer inputs as follows.
\begin{definition}[Accessibility relation for robustness]\label{def:robustness:relation}\rm
We define an accessibility relation $\Rrob \subseteq \calw \times \calw$ by:
\[
\Rrob \eqdef
\left\{ (w, w') \in \calw\times\calw \,\mid\, \Winfu(\sigma_{w}(x),\, \sigma_{w'}(x)) \le \varepsilon
\right\}\!{,}
\]
where $\Winfu$ is $\infty$-Wasserstein distance 
w.r.t. a metric $\utmetric$ 
in Definition~\ref{def:p-Wasserstein-metric}.
\end{definition}

Then $ (w, w') \in \Rrob$ represents that the two distributions $\sigma_{w}(x)$ and $\sigma_{w'}(x)$ of inputs to the classifier $C$ are close in terms of the distance~$\Winfu$.%
\footnote{$\Winfu(\sigma_{w}(x),\, \sigma_{w'}(x)) \le \varepsilon$ expresses that each value of the input $x$ from the dataset $w$ is close to the corresponding value of $x$ from $w'$ in terms of the metric $\utmetric$ between individual data.
For example, each input image $x$ in the dataset $w$ looks similar to the corresponding image in $w'$ from the human' eyes.}
Intuitively, for example, $\Winfu$ means the distance between two image datasets $\sigma_{w}(x)$ and $\sigma_{w'}(x)$ when the distance between individual images are measured by a metric~$\utmetric$.

Then an epistemic formula $\MKeD \phi$ represents that we are confidence that $\phi$ is true even when the input data are perturbed by noise of the level~$\varepsilon$ or smaller.

\subsection{Probabilistic Robustness against Targeted Attacks}
\label{sub:target-robustness}

When a robustness attack aims at misclassifying an input as a specific target label $\ellt$, then it is called a \emph{targeted attack}.
For instance, in the above-mentioned attack by~\cite{Goodfellow:ICLR:15}, a gibbon is the target into which a panda's photo is misclassified.

In this section, we discuss how we formalize robustness using the epistemic operator $\MKeD$.
We denote by $v\in\cald$ an original input image in the test dataset $\wtest$, and by $\tiv\in\cald$ an image obtained by perturbing the original image $v$ by noise.

A first definition of robustness against targeted attacks might be:
\begin{quote}
For any $v, \tiv \in \cald$,\,
if $H(v) = {\sf panda} \mbox{ and } \utmetric(v, \tiv) \le~\varepsilon$,
then $C(v') \neq {\sf gibbon}$,
\end{quote}
which represents that when an image $\tiv$ is obtained by perturbing a panda's photo $v$ by noise, then it will not be classified as the target label {\sf gibbon} at all.
This can be formalized using \StatEL{} by:
\[
\M, \wtest \models h_{\sf panda}(x) \erightarrow 
\MKeD \PR{0} \psi_{\sf gibbon}(x)
{.}
\]
However, this notion does not accept a negligible probability of misclassification,
and does not cover the case where the human cannot recognize the perturbed image $\tiv$ as {\sf panda}
(e.g., when the perturbed image $\tiv$ is obtained by linear displacement, rescaling, and rotation~\cite{Athalye:18:ICML}, then $H(\tiv) \neq {\sf panda}$ may hold).

To overcome these issues, we introduce the following definition with some conditional probability $\delta$ of misclassification as follows.

\begin{definition}[Targeted robustness]\label{def:robust:target}\rm
Let $\delta\in[0, 1]$.
Given a dataset $\wtest$, a classifier $C$ satisfies \emph{probabilistic targeted robustness} w.r.t. an actual label $\ell$ and a target label $\ellt$
if for any input $v \in \supp(\sigma_{\wtest}(x))$ from the dataset $\wtest$, and for any perturbed input $\tiv \in \cald$ s.t. $\utmetric(v, v') \le \varepsilon$, we have:
\begin{align}\label{eq:PTR}
\Pr[\, C(\tiv) =  \ellt \mid H(\tiv) = \ell \,] \le \delta
{.}
\end{align}
\end{definition}

For instance, when 
the actual class $\ell$ is ${\sf panda}$ and 
the target label $\ellt$ is ${\sf gibbon}$, then the classifier $C$ misclassifies a panda's photo as ${\sf gibbon}$ with only a small probability~$\delta$.

Now we express this robustness notion with $I = [1-\delta, 1]$ by using $\StatEL$.

\begin{restatable}[Targeted robustness]{prop}{PTrobust}
\label{prop:PT-robust}
Let $I \subseteq [0, 1]$.
The probabilistic targeted robustness w.r.t. an actual label $\ell$ and a target label $\ellt$ under a given test dataset $\wtest$ is expressed by
$\,\wtest \models \TargetRobust_{\ell, \ellt, I}(x)$ where:
\begin{align*}
\TargetRobust_{\ell, \ellt, I}(x) \eqdef\,
\MKeD \bigl( h_{\ell}(x) \erightarrow \PR{I} \neg\, \psi_{\ellt}(x) \bigr).
\end{align*}
\end{restatable}

\begin{proof}
Let $w'$ be a possible world such that $(\wtest, w')\in\Rrob$.
Then $w'$ corresponds to the dataset obtained by perturbing each data in $w$.
Let $\tiv \in \supp(\sigma_{w'}(x))$.
Then $\tiv$ represents a perturbed input.
Let $w'' = w'|_{h_{\ell}(x)}$.
Then \eqref{eq:PTR} is logically equivalent to
$w'' \models \PR{[0, \delta]} \psi_{\ellt}(x)$.
By Definition~\ref{def:modal:conditioning}, 
$w' \models h_{\ell}(x) \erightarrow \PR{[0, \delta]} \psi_{\ellt}(x)$.
By $I = [1-\delta, 1]$,\,
$w' \models h_{\ell}(x) \erightarrow \PR{I} \neg\, \psi_{\ellt}(x)$.
Therefore this proposition follows from the semantics for $\MKeD$.
\qed
\end{proof}

Since the $L^p$-distances\footnote{\modified{The $L^p$-distance between $n$-dimensional real vectors $x$ and $x'$ is written $\| x - x' \|_p$ where the $p$-norm is defined by 
$\| v \|_p = (\sum_{i=1}^{n} |v_i|^p)^{1/p}$.}}\,%
are often regarded as reasonable approximations of human perceptual distances~\cite{Carlini17:SP}, they are used as distance constraints on the perturbation in many researches on targeted attacks (e.g.~\cite{Szegedy:14:ICLR,Goodfellow:ICLR:15,Carlini17:SP}).
Our model can represent the robustness against these attacks by using the $L^p$-distance as a metric~$\utmetric$ for $\Rrob$.

\begin{figure*}[t]
\centering
\begin{tikzpicture}
\draw [black, very thin, name path=rectangle, rotate=0] (-0.5,-0.5) rectangle +(9.5,4.5);
\fill [gray!15!white] (0.0,2.85) rectangle (8.00,3.75);
\draw [black, dotted, name path=rectangle, rotate=0] (0.0,2.85) rectangle +(8.00,0.90);
\coordinate (W1) at (0.1,3.55) node at (W1) [right] {{\scriptsize Prob. non-targeted robustness (Sect.~\ref{sub:total-robustness})}};
\coordinate (E1) at (0.1,3.15) node at (E1) [right] {
$\TargetRobust_{\ell, \ellt, I}(x) \eqdef\,
\MKeD \bigl( h_{\ell}(x) \erightarrow \PR{I} \neg\, \psi_{\ellt}(x) \bigr)$};
\fill [gray!15!white] (0.0,1.35) rectangle (8.00,2.25);
\draw [black, dotted, name path=rectangle, rotate=0] (0.0,1.35) rectangle +(8.00,0.90);
\coordinate (W2) at (0.1,2.05) node at (W2) [right] {{\scriptsize Prob. targeted robustness (Sect.~\ref{sub:target-robustness})}};
\coordinate (E2) at (0.74,1.65) node at (E2) [right] {
$\TotalRobust_{\ell, I}(x) \eqdef\,
\MKeD \bigl( h_{\ell}(x) \erightarrow \PR{I} \psi_\ell(x) \bigr)$};
\fill [gray!15!white] (0.0,-0.15) rectangle (8.00,0.75);
\draw [black, dotted, name path=rectangle, rotate=0] (0.0,-0.15) rectangle +(8.00,0.90);
\coordinate (W3) at (0.1,0.5) node at (W3) [right] {{\scriptsize Recall (Sect.~\ref{sub:ML:correctness:measures})}};
\coordinate (E3) at (0.87,0.15) node at (E3) [right] { 
$\Recall_{\ell, I}(x) \eqdef\,
h_{\ell}(x) \erightarrow \PR{I} \psi_\ell(x)$};
\coordinate (CompL1b) at (2.2,2.8) node at (CompL1b) {};
\coordinate (CompL2t) at (2.2,2.3) node at (CompL2t) {};
\draw [->,double] (CompL1b)--(CompL2t);
\coordinate (Relation2) at (3.6,2.5) node at (Relation2) {\scriptsize Proposition~\ref{prop:relation-robust} (1)};
\coordinate (CompL2b) at (2.2,1.3) node at (CompL2b) {};
\coordinate (CompL3t) at (2.2,0.8) node at (CompL3t) {};
\draw [->,double] (CompL2b)--(CompL3t);
\coordinate (Relation3) at (3.6,1.0) node at (Relation3) {\scriptsize Proposition~\ref{prop:relation-robust} (2)};
\end{tikzpicture}
\caption{
Robustness notions and their relationships.
\label{fig:compare-robustness}}
\end{figure*}
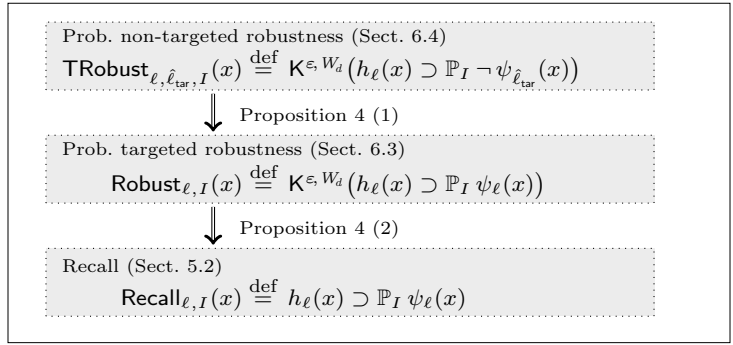

\subsection{Probabilistic Robustness against Non-Targeted Attacks}
\label{sub:total-robustness}

In this section we formalize \emph{non-targeted attacks}~\cite{Moosavi:16:CVPR,Madry:18:ICLR} in which adversaries try to misclassify inputs as some arbitrary incorrect labels (i.e., not as a specific label like a gibbon). 
Compared to targeted attacks, this kind of attacks are easier to mount, but harder to defend.

We first define the notion of robustness against non-targeted attacks as follows.

\begin{definition}[Non-targeted robustness]\label{def:robust:non-target}\rm
Let $\delta\in[0, 1]$.
Given a dataset $\wtest$, a classifier $C$ satisfies \emph{probabilistic non-targeted robustness} w.r.t. an actual label $\ell$
if for any input $v \in \supp(\sigma_{\wtest}(x))$ from the dataset $\wtest$, and for any perturbed input $\tiv \in \cald$ s.t. $\utmetric(v, v') \le \varepsilon$, we have:
\begin{align*}
\Pr[\, C(\tiv) =  \ell \mid H(\tiv) = \ell \,] >  1 - \delta
{.}
\end{align*}
\end{definition}

Now we express this robustness notion with $I = [1-\delta, 1]$ by using $\StatEL$.

\begin{restatable}[Non-targeted robustness]{prop}{PNTrobust}
\label{prop:PNT-robust}
Let $I \subseteq [0, 1]$.
The probabilistic non-targeted robustness under a test dataset $\wtest$ is expressed by
$\,\wtest \models \TotalRobust_{\ell, I}(x)$ where:
\begin{align*} 
\TotalRobust_{\ell, I}(x) 
&\eqdef\,
\MKeD \bigl( h_{\ell}(x) \erightarrow \PR{I} \psi_\ell(x) \bigr)
\\ &=
\MKeD \Recall_{\ell, I}(x)
{.}
\end{align*}
\end{restatable}

\begin{proof}
The proof is analogous to that for Proposition~\ref{prop:PT-robust}.
\qed
\end{proof}

\subsection{Relationships among Robustness Notions}
\label{sub:related-robustness}

In this section we present relationships among notions of robustness and performance, and discuss properties related to robustness.

We first present the following proposition immediate from the definitions.

\begin{restatable}[Relationships among notions]{prop}{RelationRobust}
\label{prop:relation-robust}
\!Let $I \subseteq [0, 1]$ and $\ell, \ellt \in\!\Label$.
Then we have:
\begin{enumerate}
\item 
\!$\wtest \models \TotalRobust_{\ell\!, I}(x)$ implies 
$\wtest \models \TargetRobust_{\ell\!, \ellt\!, I}(x)$.
\item 
\!$\wtest \models \TotalRobust_{\ell, I}(x)$ implies 
$\M, \wtest \models \Recall_{\ell, I}(x)$.
\end{enumerate}
\end{restatable}

The first claim means that probabilistic non-targeted robustness is not weaker than probabilistic targeted robustness for the same $I$.
The second claim means that probabilistic non-targeted robustness implies recall without perturbation noise.
Note that this is immediate from the reflexivity of $\Rrob$.

Next we remark that our extension of \StatEL{} can be used to describe a certain situation where adversarial attacks are mitigated.
When we apply some mechanism $T$ that preprocesses a given input to mitigate attacks on robustness, then the probabilistic targeted robustness is expressed as
$\wtest \models \MTDt \TotalRobust_{\ell, I}(x)$
where $\MTDt$ is the modality for the dataset transformation $T$.

Finally, we recall that by Proposition~\ref{prop:PNT-robust}, robustness can be regarded as recall in the presence of perturbed noise.
This implies that for each property $\phi$ in the table of confusion (Table~\ref{table:confusion}), we could consider $\MKeD \phi$ as a property to evaluate the classification performance in the presence of adversarial inputs
although this has not been formalized in the literature of robustness of machine learning as far as we recognize.
For example, \emph{precision robustness} $\MKeD \Precision_{\ell,i}(x)$ represents that in the presence of perturbed noise, the prediction is correct with a probability $i$ given that it is positive.
For another example, \emph{accuracy robustness} $\MKeD \Accuracy_{\ell,i}(x)$ represents that in the presence of perturbed noise, the prediction is correct (whether it is positive or negative) with a probability $i$.

\begin{example}[Robustness of pedestrian detection]\label{eg:human:robust}
\modified{We illustrate robustness notions using the pedestrian detection in Example~\ref{eg:human:performance} in Section~\ref{sub:ML:correctness:measures}.
We deal with a binary classifier $C$ that detects whether a pedestrian is crossing the road in a photo image $x$.}

\modified{The non-targeted robustness $\MKeD \Recall_{\ell, 0.9}(x)$ represents that in the presence of perturbed noise to the input image $x$, with probability $0.9$ the classifier $C$ can detect a person crossing the road when the human can actually recognize.
This robustness is crucial for an autonomous car not to hit a pedestrian.}

\modified{The precision robustness $\MKeD \Precision_{\ell,0.9}(x)$ represents that in the presence of perturbed noise to $x$, with probability $0.9$ the human can actually recognize a person crossing the road when the classifier $C$ detects it.
This type of robustness is important for an autonomous car to avoid stopping suddenly due to a false alarm (not take the crash from the car behind).}
\end{example}

\section{Formalizing the Fairness of Classifiers}
\label{sec:ML:fairness}
Many studies have proposed and investigated various notions of fairness in machine learning~\cite{Barocas:19:book}.
Informally, these fairness notions mean that the results of machine learning tasks are irrelevant of some sensitive attributes, e.g., gender, age, race, disease, political/religious view.
In a recently few years, there have been studies on the testing methods for fairness of machine learning~\cite{Galhotra:17:ESECFSE,Angell:18:ESECFSE,Udeshi:18:ASE}.

In this section, we formalize popular notions of fairness of supervised learning by using our extension of \StatEL{}.
Here we focus on the fairness that should be maintained in the \emph{impact} (i.e., the results of machine learning tasks) rather than the \emph{treatment} (i.e., the process of machine learning tasks).
This is because previous research show that many seemingly neutral features have statistical relationships with sensitive attributes, and hence just ignoring or removing sensitive attributes in the process of data preparation and training\footnote{Such \emph{unawareness} requires that sensitive attributes are not explicitly used in the learning process. However, \StatEL{} may not be suited to formalizing this requirement.} is often ineffective or harmful to achieve the fairness and performance of learning tasks.

\subsection{Basic Ideas and Notations}
\label{sec:basic:fairness}

Various notions of fairness in supervised learning are classified into three categories: \emph{independence}, \emph{separation}, and \emph{sufficiency}~\cite{Barocas:19:book}.
All of these have the form of (conditional) independence or its relaxation,
and thus can be formalized using the modal operator $\IND{x}{\varepsilon,D}$ for conditional indistinguishability (defined in Sect.~\ref{sub:modal:indistinguishable}) in our extension of \StatEL{}.\footnote{Compared to the preliminary version~\cite{Kawamoto:19:SEFM} of this paper, we corrected errors and changed the formalization into a more comprehensible form by introducing the operator $\IND{x}{\varepsilon,D}$ and by removing the \emph{counter factual epistemic operators} and a formula $\xi_d$ representing that the input is drawn from a dataset $d$.}

In the formalization of fairness notions, we use a distributional Kripke model $\M =(\calw, (\calr_a)_{a\in\cala}, \allowbreak (V_s)_{s\in\cals})$.
Recall that $x$, $y$, and $\hy$ are measurement variables respectively denoting the input datum, the actual class label (given by the oracle $H$), and the predicted label (output by the classifier $C$).
Given a real world $\wtest$ (corresponding to a given test dataset), 
$\sigma_{\wtest}(x)$ is the probability distribution of $C$'s test input over $\cald$, 
$\sigma_{\wtest}(y)$ is the distribution of the actual label over $\Label$,
and $\sigma_{\wtest}(\hy)$ is the distribution of $C$'s output over $\Label$.

Fairness notions are usually defined in terms of some \emph{sensitive attribute} (e.g., gender, age, race, disease, political/religious view), which is defined as a tuple of subsets of the input data domain $\cald$.
For example, a sensitive attribute based on ages can be defined as a pair of groups $G_0$ (input data with ages 21 to 60) and $G_1$ (ages 61 to 100).
For each group $G\subseteq\cald$ of inputs, we introduce a static formula $\eta_{G}(x)$ representing that an input $x$ belongs to~$G$.
Formally, this is interpreted~by:
\begin{align*}
\mbox{For each state $s\in\cals$, }~ 
s\models \eta_{G}(x) ~\mbox{ iff }~ \sigma_{s}(x) \in G.
\end{align*}

Roughly speaking, a machine learning task is said to be fair if the performance of the task for a group $G_0$'s input is similar to that for another group $G_1$'s input.%
\footnote{Some fairness notions (e.g., equal opportunity) assume $G_1 = \cald \setminus G_0$.}
In the following sections, we formalize the three categories of fairness of classifiers and their relaxation.
\modified{A summary of this formalization is presented in Table~\ref{table:various:fairness}.}

\begin{table*}[t]
\caption{Popular notions of fairness of machine learning \label{table:various:fairness}}
\centering
\begin{tabular}{|l|l|} \hline
&\\[-1.9ex]
\hspace{-1ex}Sect.\hspace{-1.3ex}\!& Formalization of fairness notions \\[0.1ex] \hline
&\\[-1.2ex]
\!\ref{sub:fair:independence}\! &
\!Independence (a.k.a. group fairness) \\[0.1ex]
&
\!$\GrpFair_{\varepsilon}(x, \hy) \eqdef
\bigl( \eta_{G_0}(x) \wedge \psi(x, \hy) \bigr)
\IND{\hy}{\varepsilon,\Dtv}
\bigl( \eta_{G_1}(x) \wedge \psi(x, \hy) \bigr)$
\\[1.0ex] \hline
&\\[-1.2ex]
\!\ref{sub:fair:seperation}\! &
\!Separation (a.k.a. equalized odds) \\[0.1ex]
&
\!$\EqOdds_{\varepsilon}(x, \hy) \eqdef
{\displaystyle \bigwedge_{\ell\in\Label}}
\Bigl(
\bigl( \eta_{G_0}(x) \wedge \psi(x, \hy) \wedge h_{\ell}(x) \bigr)
\IND{\hy}{\varepsilon,\Dtv}
\bigl( \eta_{G_1}(x) \wedge \psi(x, \hy) \wedge h_{\ell}(x) \bigr)
\Bigr)$\hspace{-1ex}\!%
\\[1.0ex] \hline
&\\[-1.2ex]
\!\ref{sub:fair:seperation}\! &
\!Equal opportunity (a relaxation of separation) \\[0.1ex]
&
\!$\EqOpp(x, \hy) \eqdef
\bigl( \eta_{G_0}(x) \wedge \psi(x, \hy) \wedge h_{\ell}(x) \bigr)
\IND{\hy}{0,\Dtv}
\bigl( \neg\eta_{G_0}(x) \wedge \psi(x, \hy) \wedge h_{\ell}(x) \bigr)$
\\[1.0ex] \hline
&\\[-1.2ex]
\!\ref{sub:fair:sufficiency}\! &
\!Sufficiency (a.k.a. conditional use accuracy equality) \\[0.1ex]
&
\!$\Suffice_{\varepsilon}(x, y) \eqdef
{\displaystyle \bigwedge_{\hell\in\Label}}
\Bigl(\!
\bigl( \eta_{G_0}(x) \wedge \psi_{\hell}(x) \wedge h(x, y) \bigr)
{\,\IND{y}{\varepsilon,\Dtv}}
\bigl( \eta_{G_1}(x) \wedge \psi_{\hell}(x) \wedge h(x, y) \bigr)
\!\Bigr)$\hspace{-1ex}\!%
\\[1.0ex] \hline
  \end{tabular}
\end{table*}

\subsection{Independence  (a.k.a. Group Fairness, Statistical Parity) and its Relaxation}
\label{sub:fair:independence}

In this section we explain and formalize the notion of \emph{independence}~\cite{Calders:10:DMKD},
which is also known as \emph{group fairness}~\cite{Dwork:12:ITCS}%
\,\footnote{
In previous literature, independence has been referred to also as different terminologies, such as \emph{statistical parity}, \emph{demographic parity}, and \emph{disparate impact}.}, 
and its relaxed notion.
Intuitively, independence means that the predicted label $\hy$ does not have statistical relationships with the membership in a sensitive group.
For example, independence does not allow a bank's lending rate to be correlated with a sensitive attribute such as gender.

We first present the definition of a relaxed notion of independence, called \emph{group fairness up to bias $\varepsilon$}~\cite{Dwork:12:ITCS} as follows.
Intuitively, this is the property that the output distributions of the classifier are roughly identical when input data belong to different groups.

Formally, this fairness notion is defined as follows.
\begin{definition}[Independence, group fairness]\label{def:independence} \rm
Let $G_0, G_1 \subseteq \cald$ be sets of input data constituting a sensitive attribute.
For each $b = 0, 1$, let 
$\mu_{G_b}\in\Dists\Label$ be the probability distribution of the predicted label $\hell$ output by a classifier $C$ when an input $v$ is sampled from a test dataset $\wtest$ and belongs to $G_b$; i.e., for each $\hell\in\Label$,
\begin{align}\label{eq:group-fairness}
\mu_{G_b}[\hell\,]
&\eqdef
\Pr[\, C(v) = \hell \,\,|\, v \randassign \sigma_{\wtest}(x) \mbox{ and } v \in G_b  \,]
{.}
\end{align}
Then a classifier $C$ satisfies the \emph{group fairness between groups $G_0$ and $G_1$ up to bias} $\varepsilon$ 
if $\Dtv( \mu_{G_0} \| \mu_{G_1} ) \leq \varepsilon$,
where $\Dtv$ is the total variation between distributions (defined in 
Sect.~\ref{sub:stat-distance}).
A classifier $C$ satisfies \emph{independence} w.r.t. groups $G_0$ and $G_1$ if it satisfies the group fairness between $G_0$ and $G_1$ up to bias~$0$.
\end{definition}

Now we express this fairness notion using our extension of \StatEL{} as follows.

\begin{restatable}[Independence,\! group fairness]{prop}{Independence}
\label{prop:independence}
\!The group fairness between groups $G_0$ and $G_1$ up to bias $\varepsilon$ 
under a given test dataset $\wtest$
is expressed as 
$\wtest \models \GrpFair_{\varepsilon}(x, \hy)$ where:
\[
\GrpFair_{\!\varepsilon}(x, \hy)\!\eqdef\!
\bigl( \eta_{G_{\!0}}\!(x) \wedge \psi(x, \hy) \bigr)
\IND{\hy}{\varepsilon,\Dtv}
\!\bigl( \eta_{G_{\!1}}\!(x) \wedge \psi(x, \hy) \bigr)
{.}
\]
Independence (without bias $\varepsilon$) is expressed by
$\wtest \models \GrpFair_{0}(x, \hy)$.
\end{restatable}

\begin{proof}
Let $w_b = \wtest|_{\eta_{G_b}(x) \wedge \psi(x, \hy)}$.
It follows from \eqref{eq:group-fairness} that for each $\hell\in\Label$,
$
\mu_{G_b}[\hell\,] = \Pr[\, \sigma_{s}(\hy) = \hell \,\mid\, s \randassign w_b \,],
$
hence $\mu_{G_b} = \sigma_{w_b}(\hy)$.
Thus, by Definition~\ref{def:independence}, the group fairness between groups $G_0$ and $G_1$ up to bias $\varepsilon$ is given by 
$\Dtv( \sigma_{w_0}(\hy) \| \sigma_{w_1}(\hy) ) \leq \varepsilon$.
Therefore, this proposition follows from Lemma~\ref{lem:IND-divergence}.
\qed
\end{proof}

\begin{example}[Independence in pedestrian detection]\label{eg:human:Ind}
\modified{We illustrate independence using the pedestrian detection in Example~\ref{eg:human:performance} in Section~\ref{sub:ML:correctness:measures}.
We deal with a binary classifier $C$ that detects whether or not a pedestrian is crossing the road in an image $x$.
We write $\man(x)$ (resp. $\woman(x)$) to represent that an image $x$ includes a man (resp. woman) that may or not be crossing the road.
Let $\psi(x, \hy)$ 
represent that given an input image $x$, the classifier $C$ 
returns $\hy$ (that is either the detection of a person crossing the road or not).}

\modified{Then the independence between men and women
$
\GrpFair_{\!0}(x, \hy) {\eqdef}
\bigl( \man(x) \wedge \psi(x, \hy) \bigr)
\IND{\hy}{0,\Dtv}
\bigl( \woman(x) \wedge \psi(x, \hy) \bigr)
$
means that the probability of detecting a pedestrian crossing the road is the same between men and women.
This fairness guarantees that men and women are equally detectable as pedestrians, hence equally safe against an autonomous car.
Here independence does not rely on the actual label $y$, i.e., on whether there is a pedestrian crossing the road that can be detected by human eyes.}
\end{example}

\subsection{Separation (a.k.a. Equalized Odds) and its Relaxation (Equal Opportunity)}
\label{sub:fair:seperation}

In this section we explain and formalize the notion of \emph{separation}~\cite{Barocas:19:book}%
\,\footnote{In previous literature, separation has been referred to also as 
\emph{disparate mistreatment}~\cite{Zafar:17:WWW} and \emph{conditional procedure accuracy equality}~\cite{Berk:18:SMR}.}, 
which is well-known as \emph{equalized odds}~\cite{Hardt:16:NIPS}, 
and its relaxed notion called \emph{equal opportunity}~\cite{Hardt:16:NIPS}.
The motivation behind these notions is to capture typical scenarios in which sensitive characteristics may have statistical relationships with the actual class label.
For instance, even when some sensitive attribute is correlated with an actual default rate on loans, banks might want to have a different lending rate for people who have a higher default rate.
However, independence (group fairness) does not allow this, since it requires that the lending rate should be statistically independent of the sensitive attribute.

To overcome this problem, the notion of separation allows statistical relationships between a sensitive attribute and the predicted label $\hy$ output by the classifier $C$ to the extent that this is justified by the actual class label $y$.
More precisely, separation means that the predicted label $\hy$ is conditionally independent of the membership in a sensitive group, given an actual class label $y$.

Formally, separation 
is defined as a property that recall (true positive rate) and specificity (true negative rate, explained in Table~\ref{table:confusion}) are the same for all the groups,
and equal opportunity is defined as a special case of separation only for an advantageous class label.
\begin{definition}[Separation \& equal opportunity]\label{def:separation} \rm
Given a group $G_b \subseteq \cald$ and an actual class label $\ell$, let $\mu_{G_b,\ell}\in\Dists\Label$ be the probability distribution of the predicted label $\hell$ output by a classifier $C$ when an input $v\in G_b$ is sampled from a test dataset $\wtest$ and is associated with an actual label $\ell$; i.e., for each $\hell\in\Label$,
\begin{align}\label{eq:equalized-odds}
\mu_{G_b,\ell}[\hell\,]
&\!\eqdef
\Pr[\, C(v) = \hell \,|\, 
v {\randassign} \sigma_{\wtest}(x),  v \in G_b, H(v)\!= \ell  \,\,]
{.}
\end{align}
A classifier $C$ satisfies \emph{separation} between two groups $G_0$ and $G_1$ 
if $\mu_{G_0,\ell} = \mu_{G_1,\ell}$ holds for all $\ell\in\Label$.
A classifier $C$ satisfies \emph{equal opportunity} of an advantageous label $\ell$ w.r.t. a group $G_0$
if $\mu_{G_0,\ell} = \mu_{G_1,\ell}$ where $G_1 = \cald \setminus G_0$.
\end{definition}

Now we express these two notions using our extension of \StatEL{} as follows.

\begin{restatable}[Separation]{prop}{Separation}
\label{prop:separation}
Let $\gamma(x, \ell, \hy) \eqdef \psi(x, \hy) \allowbreak \wedge h_{\ell}(x)$.
The separation between two groups $G_0$ and $G_1$ 
under a given test dataset $\wtest$
is expressed as 
$\wtest \models \EqOdds_{0}(x, \hy)$ where: 
\begin{align*}
&
\EqOdds_{\varepsilon}(x, \hy) \eqdef
\\[-0.5ex]&
\bigwedge_{\ell\in\Label}
\Bigl(
\bigl( \eta_{G_0}(x) \wedge \gamma(x, \ell, \hy) \bigr)
\IND{\hy}{\varepsilon,\Dtv}
\bigl( \eta_{G_1}(x) \wedge \gamma(x, \ell, \hy) \bigr)
\Bigr)
{.}
\end{align*}
\end{restatable}

\begin{proof}
Let $\ell\in\Label$ and $w_{b,\ell} = \wtest|_{\eta_{G_b}(x) \wedge \psi(x, \hy) \wedge h_{\ell}(x)}$.
It follows from \eqref{eq:equalized-odds} that:
\begin{align*}
\mu_{G_b,\ell}[\hell\,] = \Pr[\, \sigma_{s}(\hy) = \hell \,\mid\, s \randassign w_{b,\ell} \,],
\end{align*}
hence $\mu_{G_b,\ell} = \sigma_{w_{b,\ell}}(\hy)$.
Thus, by Definition~\ref{def:separation}, the separation between $G_0$ and $G_1$ is given by $\sigma_{w_{0,\ell}}(\hy) = \sigma_{w_{1,\ell}}(\hy)$ for all $\ell\in\Label$.
Therefore, this proposition
follows from Proposition~\ref{prop:conditional-identity}.
\qed
\end{proof}

It should be noted that for $\varepsilon > 0$,\, $\EqOdds_{\varepsilon}(x, \hy)$ represents a relaxation of separation up to bias $\varepsilon$ in terms of total variation $\Dtv$.

\begin{restatable}[Equal opportunity]{prop}{EqualOpportunity}
\label{prop:equal-opportunity}
\!Let $\gamma(x, \ell, \hy) {\eqdef} \allowbreak \psi(x, \hy) \wedge h_{\ell}(x)$.
The equal opportunity of a label $\ell$ w.r.t. a group $G_0$ 
under a given test dataset $\wtest$
is expressed as 
$\wtest \models \EqOpp(x, \hy)$ where:
\begin{align*}
& \EqOpp(x, \hy) \eqdef
\\[-0.5ex]&
\bigl( \eta_{G_0}(x) \wedge \gamma(x, \ell, \hy) \bigr)
\IND{\hy}{0,\Dtv}
\bigl( \neg\eta_{G_0}(x) \wedge \gamma(x, \ell, \hy) \bigr)
{.}
\end{align*}
\end{restatable}

\begin{proof}
The proof of this proposition is similar to that of Proposition~\ref{prop:separation}.
Let $G_1 = \cald \setminus G_0$.
By $\mu_{G_b,\ell} = \sigma_{w_{b,\ell}}(\hy)$,
the equal opportunity of $\ell$ w.r.t. $G_0$ is given by $\sigma_{w_{0,\ell}}(\hy) = \sigma_{w_{1,\ell}}(\hy)$.
Therefore, this proposition 
follows from Proposition~\ref{prop:conditional-identity}.
\qed
\end{proof}

\begin{example}[Separation in pedestrian detection]\label{eg:human:speparate}
\modified{We illustrate separation using the pedestrian detection in Example~\ref{eg:human:Ind}
where a binary classifier $C$ detects whether a pedestrian is crossing the road in an image $x$.
Let $\psi(x, \hy)$ (resp. $h(x, y)$) 
represent that given an image $x$,\, the classifier $C$ (resp. human) 
returns $\hy$ (resp.~$y$) representing either detection or not.}

\modified{The level of the inherent technical difficulty of detecting a female pedestrian may be different from that of a male pedestrian, because, for example, the physical appearance may tend to be different between women and men.
If we take this possible difference into account, separation can be suited instead of independence.}

\modified{The separation $\EqOdds_{0}(x, \hy)$ between men and women guarantees that the conditional probability of detecting a pedestrian crossing the road when the human can actually recognize it, is the same between men and women.
This fairness implies that (from the viewpoint of a pedestrian crossing the road) male and female pedestrians may be hit by an autonomous car as fairly as by the human-driven car.}
\end{example}

\subsection{Sufficiency (a.k.a. Conditional Use Accuracy Equality)}
\label{sub:fair:sufficiency}

In this section we explain and formalize the notion of \emph{sufficiency}~\cite{Barocas:19:book}, which is also known as \emph{conditional use accuracy equality}~\cite{Berk:18:SMR}.

While separation guarantees the equality of recall among different groups, sufficiency requires the equality of precision.
More precisely, sufficiency is defined as the property that precision (positive predictive value) and negative predictive value (presented as NPV in Table~\ref{table:confusion}) are the same for all the groups as follows.
\begin{definition}[Sufficiency]\label{def:sufficiency} \rm
Given a group $G_b \subseteq \cald$ and a predicted label $\hell$, let $\mu_{G_b,\hell}\in\Dists\Label$ be the probability distribution of the actual class label $\ell$ when an input $v\in G_b$ is sampled from a test dataset $\wtest$ and the classifier $C$ outputs the predicted label $\hell$; 
i.e., for each $\ell\in\Label$,
\begin{align}\label{eq:sufficiency}
\mu_{G_b,\hell}[\ell\,]
&\!\eqdef
\Pr[\, H(v) = \ell \,|\, v {\randassign} \sigma_{\wtest}(x),  v \in G_b,
C(v) \!= \hell  \,\,]
{.}
\end{align}
A classifier $C$ satisfies \emph{sufficiency} between two groups $G_0$ and $G_1$ 
if $\mu_{G_0,\hell} = \mu_{G_1,\hell}$ holds for all $\hell\in\Label$.
\end{definition}

Then this notion can be expressed using our extension of \StatEL{} as follows.

\begin{restatable}[Sufficiency]{prop}{Sufficiency}
\label{prop:sufficiency}
Let $\gamma'(x,y,\hell) \eqdef \psi_{\hell}(x) \allowbreak \wedge h(x, y)$.
The sufficiency between two groups $G_0$ and $G_1$ 
under a given test dataset $\wtest$
is expressed as 
$\wtest \models \Suffice_{0}(x, y)$
where:
\begin{align*}
& \Suffice_{\varepsilon}(x, y) \eqdef
\\[-0.5ex] &
\bigwedge_{\hell\in\Label}
\Bigl(
\bigl( \eta_{G_0}(x) \wedge \gamma'(x,y,\hell) \bigr)
\IND{y}{\varepsilon,\Dtv}
\bigl( \eta_{G_1}(x) \wedge \gamma'(x,y,\hell) \bigr)
\Bigr)
{.}
\end{align*}
\end{restatable}

\begin{proof}
Let $\hell\in\Label$ and $w_{b,\hell} = \wtest|_{\eta_{G_b}(x) \wedge \psi_{\hell}(x) \wedge h(x, y)}$.
It follows from \eqref{eq:sufficiency} that:
\begin{align*}
\mu_{G_b,\hell}[\ell\,] = \Pr[\, \sigma_{s}(y) = \ell \,\mid\, s \randassign w_{b,\hell} \,],
\end{align*}
hence $\mu_{G_b,\hell} = \sigma_{w_{b,\hell}}(y)$.
Thus, by Definition~\ref{def:sufficiency}, the sufficiency between $G_0$ and $G_1$ is given by $\sigma_{w_{0,\hell}}(y) = \sigma_{w_{1,\hell}}(y)$ for all $\hell\in\Label$.
Therefore, this proposition
follows from Proposition~\ref{prop:conditional-identity}.
\qed
\end{proof}

It should be noted that for $\varepsilon > 0$,\, $\Suffice_{\varepsilon}(x, y)$ represents a relaxation of sufficiency up to bias $\varepsilon$ in terms of total variation $\Dtv$.

\begin{example}[Sufficiency in pedestrian detection]\label{eg:human:sufficient}
\modified{We illustrate sufficiency using the pedestrian detection in Example~\ref{eg:human:Ind} where a classifier $C$ detects whether a pedestrian is crossing the road in an image $x$.
As mentioned in Example~\ref{eg:human:speparate}, the level of the inherent technical difficulty of detecting a male pedestrian may be different from that of a female pedestrian.
Whereas separation guarantees the equality of recall between men and women, sufficiency guarantees that of precision.}

\modified{
The sufficiency $\Suffice_{0}(x, y)$ between men and women implies that the conditional probability that there is no pedestrian crossing the road when $C$ detects it, is the same between men and women.
From the viewpoint of the car driver, when $C$ raises a false alarm and stops the car suddenly, we have no bias about which of men and women are more likely to trigger false alarms and to be blamed for that.}
\end{example}

\section{Related Work}
\label{sec:related}
In this section, we provide a brief overview of related work on the specification of statistical machine learning and on epistemic logic for describing specification.

\paragraph{Desirable properties of statistical machine learning.}

There have been a large number of papers on attacks and defences for deep neural networks~\cite{Szegedy:14:ICLR,Chakraborty:18:arxiv}.
Compared to them, however, not much work has been done to explore the formal specification of various properties of machine learning.
Seshia et al.~\cite{Seshia:18:ATVA} present a list of desirable properties of DNNs (deep neural networks) although most of the properties are presented informally without mathematical formulas.
As for robustness, Dreossi et al.~\cite{Dreossi:19:VNN}
propose a unifying formalization of adversarial input generation in a rigorous and organized manner, although they formalize and classify attacks (as optimization problems) rather than define the robustness notions themselves.

Concerning the fairness notions, 
Barocas et al.~\cite{Barocas:19:book} survey various fairness notions and classify them into the three categories: independence, separation, and sufficiency.
Gajane~\cite{Gajane:17:arxiv} surveys the formalization of fairness notions for machine learning and present some justification based on social science literature.

\paragraph{Epistemic logic for describing specification.}
Epistemic logic~\cite{vonWright:51:book} has been studied to represent and reason about knowledge and belief~\cite{Fagin:95:book,Halpern:03:book}, and has been applied to describe various properties of distributed systems.

The \emph{BAN logic}~\cite{Burrows:90:TOCS}, proposed by Burrows, Abadi and Needham, is a notable example of epistemic logic used to model and verify the authentication in cryptographic protocols.
To improve the formalization of protocols' behaviors, some epistemic approaches integrate process calculi~\cite{Hughes:04:JCS,Chadha:09:Forte}.

Epistemic logic has also been used to formalize and reason about privacy properties, including anonymity~\cite{Syverson:99:FM,Garcia:05:FMSE,Kawamoto:07:JSIAM},
receipt-freeness of electronic voting protocols~\cite{Jonker:06:WOTE},
and privacy policy for social network services~\cite{Pardo:14:SEFM}.
Temporal epistemic logic is used to express information flow security policies~\cite{Balliu:11:PLAS}.

Concerning the formalization of fairness notions, previous work in formal methods has modeled different kinds of fairness involving timing by using temporal logic rather than epistemic logic.
As far as we know, no previous work has formalized fairness notions of machine learning by using modal logic.

\paragraph{Formalization of statistical properties.}

In studies of philosophical logic, Lewis~\cite{Lewis:80:subjectivist} shows the idea that when a random value has various possible probability distributions, then those distributions should be represented on distinct possible worlds.
Bana~\cite{Bana:17:EPSP} puts Lewis's idea in a mathematically rigorous setting. 
Recently, a modal logic called statistical epistemic logic (\StatEL{})~\cite{Kawamoto:19:FC} 
\modified{has been} proposed and used to formalize statistical hypothesis testing and the notion of differential privacy~\cite{Dwork:06:ICALP}.

\modified{To describe statistical properties of machine learning models, this work uses \StatEL{} to formalize the probabilistically chosen input to a learning model and the non-deterministically chosen dataset.
However, we could possibly employ other logics (e.g., fuzzy logic~\cite{Zadeh:1965:IC} or Markov logic network~\cite{Richardson:06:ML}) by extending them to deal with statistical sampling and non-deterministic inputs.
Exploring the possibility of different formalization using other logics is left for future work.
}

\section{Conclusion}
\label{sec:conclude}
In this paper we proposed an epistemic approach to the modeling of supervised learning and its desirable properties.
Specifically, we employed a distributional Kripke model in which each possible world corresponds to a possible dataset and modal operators are interpreted as transformation and testing on datasets.
Then we formalized various notions of the classification performance, robustness, and fairness of statistical classifiers by using our extension of statistical epistemic logic (\StatEL{}).
In this formalization, we clarified relationships among properties of classifiers, and relevance between classification performance and robustness.

We emphasize that this is the first attempt to use epistemic models and logical formulas to describe statistical properties of machine learning, and would be a starting point to develop theories of formal specification of machine learning.

In future work, we are planning to extend our framework 
to formally reason about system-level properties of learning-based systems.
We are also interested in developing a more general framework for the formal specification of machine learning associated with testing methods,
\modified{as well as in implementing a prototype tool}.
Our future work will also include an extension of \StatEL{} to formalize unsupervised learning and reinforcement learning.

\begin{acknowledgements}
I would like to thank the reviewers for their helpful and insightful comments. 
I am also grateful to Gergei Bana for his useful comments on part of a preliminary manuscript.
\end{acknowledgements}

\bibliographystyle{spmpsci}      
\bibliography{short,short-ML,new}   

\begin{thebibliography}{10}
\providecommand{\url}[1]{{#1}}
\providecommand{\urlprefix}{URL }
\expandafter\ifx\csname urlstyle\endcsname\relax
  \providecommand{\doi}[1]{DOI~\discretionary{}{}{}#1}\else
  \providecommand{\doi}{DOI~\discretionary{}{}{}\begingroup
  \urlstyle{rm}\Url}\fi

\bibitem{Angell:18:ESECFSE}
Angell, R., Johnson, B., Brun, Y., Meliou, A.: Themis: automatically testing
  software for discrimination.
\newblock In: Proc. {ESEC/SIGSOFT} {FSE}, pp. 871--875. {ACM} (2018).
\newblock \doi{10.1145/3236024.3264590}

\bibitem{Athalye:18:ICML}
Athalye, A., Engstrom, L., Ilyas, A., Kwok, K.: Synthesizing robust adversarial
  examples.
\newblock In: Proc. {ICML}, pp. 284--293 (2018)

\bibitem{Balliu:11:PLAS}
Balliu, M., Dam, M., Guernic, G.L.: Epistemic temporal logic for information
  flow security.
\newblock In: Proc. of {PLAS}, p.~6 (2011).
\newblock \doi{10.1145/2166956.2166962}

\bibitem{Bana:17:EPSP}
Bana, G.: Models of objective chance: An analysis through examples.
\newblock In: Making it Formally Explicit, pp. 43--60. Springer International
  Publishing (2017).
\newblock \doi{10.1007/978-3-319-55486-0\_3}

\bibitem{Barocas:19:book}
Barocas, S., Hardt, M., Narayanan, A.: Fairness and Machine Learning.
\newblock fairmlbook.org (2019).
\newblock \url{http://www.fairmlbook.org}

\bibitem{Berk:18:SMR}
Berk, R., Heidari, H., Jabbari, S., Kearns, M., Roth, A.: Fairness in criminal
  justice risk assessments: The state of the art.
\newblock Sociological Methods \& Research  (2018).
\newblock \doi{10.1177/0049124118782533}

\bibitem{Blackburn:01:book}
Blackburn, P., de~Rijke, M., Venema, Y.: Modal Logic.
\newblock Cambridge Tracts in Theoretical Computer Science. Cambridge
  University Press (2001).
\newblock \doi{10.1017/CBO9781107050884}

\bibitem{Burrows:90:TOCS}
Burrows, M., Abadi, M., Needham, R.M.: A logic of authentication.
\newblock {ACM} Trans. Comput. Syst. \textbf{8}(1), 18--36 (1990).
\newblock \doi{10.1145/77648.77649}

\bibitem{Calders:10:DMKD}
Calders, T., Verwer, S.: Three naive bayes approaches for discrimination-free
  classification.
\newblock Data Min. Knowl. Discov. \textbf{21}(2), 277--292 (2010).
\newblock \doi{10.1007/s10618-010-0190-x}

\bibitem{Carlini17:SP}
Carlini, N., Wagner, D.A.: Towards evaluating the robustness of neural
  networks.
\newblock In: Prc. {S\&P}, pp. 39--57 (2017).
\newblock \doi{10.1109/SP.2017.49}

\bibitem{Chadha:09:Forte}
Chadha, R., Delaune, S., Kremer, S.: Epistemic logic for the applied pi
  calculus.
\newblock In: Proc. of {FMOODS/FORTE}, pp. 182--197 (2009).
\newblock \doi{10.1007/978-3-642-02138-1\_12}

\bibitem{Chakraborty:18:arxiv}
Chakraborty, A., Alam, M., Dey, V., Chattopadhyay, A., Mukhopadhyay, D.:
  Adversarial attacks and defences: {A} survey.
\newblock CoRR \textbf{abs/1810.00069} (2018).
\newblock \urlprefix\url{http://arxiv.org/abs/1810.00069}

\bibitem{Dreossi:19:VNN}
Dreossi, T., Ghosh, S., Sangiovanni-Vincentelli, A.L., Seshia, S.A.: A
  formalization of robustness for deep neural networks.
\newblock In: Proc. {VNN} (2019)

\bibitem{Dwork:06:ICALP}
Dwork, C.: Differential privacy.
\newblock In: Proc. of {ICALP}, pp. 1--12 (2006)

\bibitem{Dwork:12:ITCS}
Dwork, C., Hardt, M., Pitassi, T., Reingold, O., Zemel, R.S.: Fairness through
  awareness.
\newblock In: Proc. of ITCS, pp. 214--226. ACM (2012)

\bibitem{Fagin:95:book}
Fagin, R., Halpern, J., Moses, Y., Vardi, M.: Reasoning about Knowledge.
\newblock The MIT Press (1995)

\bibitem{Gajane:17:arxiv}
Gajane, P.: On formalizing fairness in prediction with machine learning.
\newblock CoRR \textbf{abs/1710.03184} (2017).
\newblock \urlprefix\url{http://arxiv.org/abs/1710.03184}

\bibitem{Galhotra:17:ESECFSE}
Galhotra, S., Brun, Y., Meliou, A.: Fairness testing: testing software for
  discrimination.
\newblock In: Proc. {ESEC/FSE}, pp. 498--510. {ACM} (2017).
\newblock \doi{10.1145/3106237.3106277}

\bibitem{Garcia:05:FMSE}
Garcia, F.D., Hasuo, I., Pieters, W., van Rossum, P.: Provable anonymity.
\newblock In: Proc. of {FMSE}, pp. 63--72 (2005).
\newblock \doi{10.1145/1103576.1103585}

\bibitem{Goodfellow:ICLR:15}
Goodfellow, I.J., Shlens, J., Szegedy, C.: Explaining and harnessing
  adversarial examples.
\newblock In: Proc. of {ICLR} (2015)

\bibitem{Halpern:03:book}
Halpern, J.Y.: Reasoning about uncertainty.
\newblock The MIT press (2003)

\bibitem{Hardt:16:NIPS}
Hardt, M., Price, E., Srebro, N.: Equality of opportunity in supervised
  learning.
\newblock In: proc. {NIPS}, pp. 3315--3323 (2016)

\bibitem{Huang:17:CAV}
Huang, X., Kwiatkowska, M., Wang, S., Wu, M.: Safety verification of deep
  neural networks.
\newblock In: Proc. {CAV}, pp. 3--29 (2017).
\newblock \doi{10.1007/978-3-319-63387-9\_1}

\bibitem{Hughes:04:JCS}
Hughes, D., Shmatikov, V.: Information hiding, anonymity and privacy: a modular
  approach.
\newblock J. of Comp. Security \textbf{12}(1), 3--36 (2004)

\bibitem{Jonker:06:WOTE}
Jonker, H.L., Pieters, W.: Receipt-freeness as a special case of anonymity in
  epistemic logic.
\newblock In: Proc.\ Workshop On Trustworthy Elections (WOTE'06) (2006)

\bibitem{Katz:17:CAV}
Katz, G., Barrett, C.W., Dill, D.L., Julian, K., Kochenderfer, M.J.: Reluplex:
  An efficient {SMT} solver for verifying deep neural networks.
\newblock In: Proc. {CAV}, pp. 97--117 (2017).
\newblock \doi{10.1007/978-3-319-63387-9\_5}

\bibitem{Kawamoto:19:FC}
Kawamoto, Y.: Statistical epistemic logic.
\newblock In: The Art of Modelling Computational Systems: {A} Journey from
  Logic and Concurrency to Security and Privacy - Essays Dedicated to Catuscia
  Palamidessi on the Occasion of Her 60th Birthday, \emph{LNCS}, vol. 11760,
  pp. 344--362. Springer (2019).
\newblock \doi{10.1007/978-3-030-31175-9\_20}

\bibitem{Kawamoto:19:SEFM}
Kawamoto, Y.: Towards logical specification of statistical machine learning.
\newblock In: Proc. {SEFM}, \emph{LNCS}, vol. 11724, pp. 293--311. Springer
  (2019).
\newblock \doi{10.1007/978-3-030-30446-1\_16}

\bibitem{Kawamoto:07:JSIAM}
Kawamoto, Y., Mano, K., Sakurada, H., Hagiya, M.: Partial knowledge of
  functions and verification of anonymity.
\newblock Transactions of the Japan Society for Industrial and Applied
  Mathematics \textbf{17}(4), 559--576 (2007).
\newblock \doi{10.11540/jsiamt.17.4\_559}

\bibitem{Kripke:63:MLQ}
Kripke, S.A.: Semantical analysis of modal logic i normal modal propositional
  calculi.
\newblock Mathematical Logic Quarterly \textbf{9}(5-6), 67--96 (1963)

\bibitem{Lewis:80:subjectivist}
Lewis, D.: A subjectivist's guide to objective chance.
\newblock In: Studies in Inductive Logic and Probability, Volume II, pp.
  263--293. Berkeley: University of California Press (1980)

\bibitem{Madry:18:ICLR}
Madry, A., Makelov, A., Schmidt, L., Tsipras, D., Vladu, A.: Towards deep
  learning models resistant to adversarial attacks.
\newblock In: Proc. {ICLR} (2018)

\bibitem{Moosavi:16:CVPR}
Moosavi{-}Dezfooli, S., Fawzi, A., Frossard, P.: Deepfool: {A} simple and
  accurate method to fool deep neural networks.
\newblock In: Proc. {CVPR}, pp. 2574--2582 (2016).
\newblock \doi{10.1109/CVPR.2016.282}

\bibitem{Pardo:14:SEFM}
Pardo, R., Schneider, G.: A formal privacy policy framework for social
  networks.
\newblock In: Proc. {SEFM}, pp. 378--392 (2014).
\newblock \doi{10.1007/978-3-319-10431-7\_30}

\bibitem{Pei:17:SOSP}
Pei, K., Cao, Y., Yang, J., Jana, S.: Deepxplore: Automated whitebox testing of
  deep learning systems.
\newblock In: Proc. {SOSP}, pp. 1--18 (2017).
\newblock \doi{10.1145/3132747.3132785}

\bibitem{Prior:1957}
Prior, A.N.: Time and modality  (1957)

\bibitem{Richardson:06:ML}
Richardson, M., Domingos, P.M.: Markov logic networks.
\newblock Mach. Learn. \textbf{62}(1-2), 107--136 (2006).
\newblock \doi{10.1007/s10994-006-5833-1}

\bibitem{Seshia:18:ATVA}
Seshia, S.A., Desai, A., Dreossi, T., Fremont, D.J., Ghosh, S., Kim, E.,
  Shivakumar, S., Vazquez{-}Chanlatte, M., Yue, X.: Formal specification for
  deep neural networks.
\newblock In: Proc. {ATVA}, pp. 20--34 (2018).
\newblock \doi{10.1007/978-3-030-01090-4\_2}

\bibitem{Syverson:99:FM}
Syverson, P.F., Stubblebine, S.G.: Group principals and the formalization of
  anonymity.
\newblock In: World Congress on Formal Methods (1), pp. 814--833 (1999).
\newblock \doi{10.1007/3-540-48119-2\_45}

\bibitem{Szegedy:14:ICLR}
Szegedy, C., Zaremba, W., Sutskever, I., Bruna, J., Erhan, D., Goodfellow,
  I.J., Fergus, R.: Intriguing properties of neural networks.
\newblock In: Proc. {ICLR} (2014)

\bibitem{Tian:18:ICSE}
Tian, Y., Pei, K., Jana, S., Ray, B.: Deeptest: automated testing of
  deep-neural-network-driven autonomous cars.
\newblock In: Proc. {ICSE}, pp. 303--314 (2018).
\newblock \doi{10.1145/3180155.3180220}

\bibitem{Udeshi:18:ASE}
Udeshi, S., Arora, P., Chattopadhyay, S.: Automated directed fairness testing.
\newblock In: Proc. {ASE}, pp. 98--108. {ACM} (2018).
\newblock \doi{10.1145/3238147.3238165}

\bibitem{Vaserstein:69:PPI}
{Vaserstein}, L.: {Markovian processes on countable space product describing
  large systems of automata.}
\newblock {Probl. Peredachi Inf.} \textbf{5}(3), 64--72 (1969)

\bibitem{vonWright:51:book}
von Wright, G.H.: An Essay in Modal Logic.
\newblock Amsterdam: North-Holland Pub. Co. (1951)

\bibitem{Zadeh:1965:IC}
Zadeh, L.: Fuzzy sets.
\newblock Information and Control \textbf{8}(3), 338 -- 353 (1965).
\newblock \doi{https://doi.org/10.1016/S0019-9958(65)90241-X}

\bibitem{Zafar:17:WWW}
Zafar, M.B., Valera, I., Gomez{-}Rodriguez, M., Gummadi, K.P.: Fairness beyond
  disparate treatment {\&} disparate impact: Learning classification without
  disparate mistreatment.
\newblock In: Proc. {WWW}, pp. 1171--1180 (2017).
\newblock \doi{10.1145/3038912.3052660}

\end{thebibliography}

\end{document}